\documentclass[12pt,preprint]{amsart}

\usepackage{amsmath, latexsym, amsfonts, amssymb, amsthm}
\usepackage{graphicx, color,hyperref,dsfont,epsfig,caption,wrapfig,subfig}
\usepackage{mathtools}
\usepackage{parskip}
\usepackage{movie15}
\hypersetup{
	linktoc=page,
	linkcolor=red,          
	citecolor=blue,        
	filecolor=blue,      
	urlcolor=cyan,
	colorlinks=true           
}

\NewDocumentEnvironment{eqs}{+b}
{\begin{equation}\begin{split}#1\end{split}\end{equation}}
{}
\numberwithin{equation}{section}

\newtheorem{theorem}{Theorem}[section]
\newtheorem{lemma}[theorem]{Lemma}
\newtheorem{proposition}[theorem]{Proposition}

\newtheorem{remark}[theorem]{Remark}
\newtheorem{defi}[theorem]{Definition}

\newcounter{conj}
\newtheorem{conjecture}[conj]{Conjecture}

\newcounter{thmc}

\theoremstyle{remark}

\theoremstyle{definition}

\renewcommand{\tilde}{\widetilde}          
\DeclareMathSymbol{\leqslant}{\mathalpha}{AMSa}{"36} 
\DeclareMathSymbol{\geqslant}{\mathalpha}{AMSa}{"3E} 
\DeclareMathSymbol{\eset}{\mathalpha}{AMSb}{"3F}     
\renewcommand{\leq}{\;\leqslant\;}                   
\renewcommand{\geq}{\;\geqslant\;}                   
\newcommand{\dd}{\text{\rm d}}             

\newcommand{\C}{\mathbb{C}}
\newcommand{\R}{\mathbb{R}}

\newcommand{\N}{\mathbb{N}}
 
\newcommand{\D}{\mathbb{D}}

\newcommand{\im}{\bm{\mathrm{i}}}

\newcommand{\E}{\mathds{E}}
\newcommand{\X}{\bm{\mathrm X}}

\newcommand{\V}{\bm{\mathrm V}}

\newcommand{\ps}[1]{\left\langle #1 \right\rangle}

\newcommand{\psrreg}[1]{\ps{#1}_{\rho}}
\newcommand{\ga}[1]{\E\left[ #1 \right]}
\newcommand{\mc}[1]{\mathcal{#1}}

\newcommand{\eqlaw}{\overset{\text{(law)}}{=}}
\def\Met{\text{Met}(\hat\C)}

\makeatletter
\newcommand{\ostar}{\mathbin{\mathpalette\make@circled*}}
\newcommand{\make@circled}[2]{%
	\ooalign{$\m@th#1\smallbigcirc{#1}$\cr\hidewidth$\m@th#1#2$\hidewidth\cr}%
}
\newcommand{\smallbigcirc}[1]{%
	\vcenter{\hbox{\scalebox{0.77778}{$\m@th#1\bigcirc$}}}%
}
\makeatother

\newcommand{\qt}[1]{\quad\text{#1}\quad}

\def\V{\bm{\mathrm V}}

\def\SET{\bm{\mathrm T}}

\def\X{\bm{\mathrm  X}}

\def\L{\bm{\mathrm L}}

\def\eps{\varepsilon}

\renewcommand{\d}{\text{\rm d}}

\def\regset{\text{Reg}}
\newcommand{\expo}[1]{\exp\left(#1\right)}

\def\eps{\varepsilon}

\def\lioureg{\Phi^g_\eps}
\def\liou{\Phi}

\def\green{G^g}

\def\i{\bm{\mathrm i}}

\def\bi{\begin{itemize}}
	\def\ei{\end{itemize}}

\def\bnum{\begin{enumerate}}
	\def\enum{\end{enumerate}}

\def\<#1{\langle #1 \rangle}

\def\Ex{\bm{\mathrm E}}

\newcommand{\norm}[1]{\left\lvert#1\right\rvert}

\newcommand{\nnorm}[1]{\left\lvert\left\lvert#1\right\rvert\right\rvert}
\newcommand{\expect}[1]{\mathbb{E}\left[#1\right]}

\newcommand{\DOZZ}{C^\mathrm{DOZZ}_\gamma}
\newcommand{\ImDOZZ}{C^{\mathrm{ImDOZZ}}_\beta}

\usepackage{bm}
\usepackage{bbm}

\title{On the local conformal structure of Imaginary Liouville theory}

\author{Baptiste Cerclé}
\address{Sorbonne Université, CNRS,
LPTHE, Paris, 75005, France}
\email{baptiste.cercle@lpthe.jussieu.fr}

\author{Romain Usciati}
\address{Université Paris-Saclay, CNRS, LPTMS, Orsay, 91405, France}
\email{romain.usciati@universite-paris-saclay.fr}
\date{\today}

\begin{document}

\begin{abstract}
    Imaginary Liouville theory was recently proposed as a path integral construction of a non-rational conformal field theory (CFT) with central charge less than one~\cite{usciatiProbabilisticConstructionNoncompactified2026}. In the present work, we establish Ward identities and Belavin-Polyakov-Zamolodchikov differential equations in this framework, assuming a precise conjecture on the decay of the Laplace transform of Imaginary Gaussian multiplicative chaos. These results provide a first step towards a rigorous implementation of the conformal bootstrap program for conformal field theories with central charge lower than one. This work is the first to investigate the intrinsic properties of Imaginary Liouville theory beyond exact computations.
\end{abstract}

\maketitle

\section{Introduction and main results}

\subsection{Imaginary Liouville theory}
\subsubsection{From Liouville theory to Imaginary Liouville theory}

Since its introduction by Polyakov in 1981 as a model for the fluctuations of random metrics on surfaces, Liouville theory has been the most prominent example of a solvable two-dimensional Conformal field theory (CFT). Physicists were able to predict exact expressions for the fundamental building blocks of Liouville theory: the \textit{structure constants} and the \textit{spectrum}~\cite{dornTwoThreepointFunctions1994}\cite{zamolodchikovConformalBootstrapLiouville1996}. Together with universal functions called conformal blocks, these quantities determine all correlation functions of the CFT. This is the essence of the conformal bootstrap program.

More recently, the mathematical contruction of Liouville theory has allowed to give a rigorous implementation of the conformal bootstrap program~\cite{guillarmouReviewProbabilisticConstruction2024a}. Correlation functions were given a probabilistic definition by interpreting Liouville path-integral in terms of the Gaussian free field (GFF)~\cite{davidLiouvilleQuantumGravity2016a}. This led to a rigorous derivation of the structure constants of Liouville theory~\cite{kupiainenIntegrabilityLiouvilleTheory2020} and culminated in a proof of the bootstrap formula expressing correlation functions of arbitrary complexity in terms of conformal blocks~\cite{guillarmouConformalBootstrapLiouville2024}.

Interestingly, several other conformal field theories appear to have many similarities with Liouville theory. One of the most remarkable examples is the relation between scaling limits of observables in certain statistical physics models and Liouville theory structure constants. More precisely, it was conjectured in the physics litterature~\cite{delfinoThreepointConnectivityTwodimensional2011}\cite{ikhlefThreepointFunctions<2016} and recently proved mathematically~\cite{angIntegrabilityConformalLoop2024} that certain natural observables arising in models of self-avoiding loops are described by the \textit{Imaginary DOZZ formula}, which is essentially the inverse of the Liouville structure constants (the DOZZ formula)\cite[Eq.~(1.3)]{angIntegrabilityConformalLoop2024}:
\begin{eqs}
\label{eq:imdozz-dozz}
    \ImDOZZ(\alpha_1,\alpha_2,\alpha_3) = \sqrt{\frac{\prod_{k=1}^3 \DOZZ(\gamma,\hat{\alpha}_k,\hat{\alpha}_k)}{\DOZZ(\gamma,\gamma,\gamma)}} \frac{1}{\DOZZ(\hat{\alpha}_1,\hat{\alpha}_2,\hat{\alpha}_3)}
\end{eqs}
where $\beta = \gamma \in (0,2)$ and $\hat{\alpha}_k = \alpha_k - \beta$. 

The Imaginary DOZZ formula was obtained independently by Schomerus~\cite{schomerusRollingTachyonsLiouville2003}, Al.~Zamolodchikov~\cite{zamolodchikovThreepointFunctionMinimal2005} and Kostov-Petkova~\cite{kostovNonrational2DQuantum2007} who extended algebraic relations
expected to be satisfied by the three-point function of a conformal field theory~\cite{teschnerLiouvilleThreepointFunction1995}. Arguably, the reason for Eq.~\eqref{eq:imdozz-dozz} to hold is that both DOZZ and Imaginary DOZZ formulae are solutions of Teschner's shift relations, albeit with real and imaginary shifts~\cite{zamolodchikovThreepointFunctionMinimal2005}. Hence both formulae arise from conformal symmetry in different regimes: the DOZZ formula describes the gravitational sector (for which the \textit{central charge} is $c_\mathrm{L} = 1 + 6\left(\gamma/2+2/\gamma\right)^2 \geq 25$) while the Imaginary DOZZ formula is expected to describe the matter sector (where the central charge is $c_{\mathrm{IL}} = 1-6\left(\beta/2-2/\beta\right)^2 \leq 1$).

\subsubsection{Exact computations in Imaginary Liouville theory} Strikingly, whether the Imaginary Liouville formula can be realized as the structure constant of a CFT admitting a well-defined bosonic path-integral formulation remained an open problem until recently.

Most work on non-compact conformal field theories at central charge less than one has focused on extending analytical relations satisfied by correlation functions in ordinary Liouville theory. For example, for theories with boundary,~\cite{bautistaBoundaryTimelikeLiouville2022} proposed some exact expressions for the bulk one-point function and the boundary two-point function by solving the Liouville shift relations with imaginary shifts. Inspired by the representation of four-point functions in ordinary Liouville theory,~\cite{ribaultLiouvilleTheoryCentral2015} proposed a \textit{bootstrap formula} at central charge less than one and provided numerical checks of its consistency.

A different line of research is closer to the field-theoretic perspective. It consists in computing correlation functions by extrapolating expressions obtained in the special regime of \textit{charge neutrality}. In that regime, the external charges parametrizing the $N$-point correlation functions satisfy $\sum_{k=1}^N \alpha_k - 2Q = -s\beta$ for some $s\in \N$. Crucially, since the work of Dotsenko-Fateev on minimal models~\cite{dotsenkoConformalAlgebraMultipoint1984}, it has been known that neutral correlation functions admit a free field representation (\textit{Coulomb gas representation}). Building on an observation already made by~\cite{kostovNonrational2DQuantum2007}, \cite{giribetTimelikeLiouvilleThreepoint2012} showed that the Imaginary DOZZ formula coincides at charge neutrality with the Coulomb gas representation. Intriguingly, for the bulk one-point function, the result obtained from shift relations~\cite{bautistaBoundaryTimelikeLiouville2022} is different from the one extrapolated from the Coulomb gas~\cite{giribetDisk1pointFunction2026}.

\subsubsection{Imaginary Liouville theory as a path integral}
In the present work, we adopt a fully constructive strategy based on a path-integral construction of the correlation functions.

An important step in this direction was made in the work of Harlow, Maltz and Witten who showed that the Imaginary DOZZ formula cannot be obtained from a conventional bosonic action~\cite[Section 7.3]{harlowAnalyticContinuationLiouville2011a}, contrary to what the so-called \textit{timelike Liouville theory} might suggest~\cite[Eq.~(2.1), (2.4)]{kostovNonrational2DQuantum2007}. Their analysis relies on the semi-classical limit $\gamma\to0$ of both the Liouville action and the DOZZ structure constants. They argue that, for any complex value of the central charge, the Liouville path integral should instead be defined by integrating over a suitable steepest-descent cycle in the space of fields, although they give no precise definition of it. More recently, the structure of such integration cycles was made explicit for a lattice discretization of Liouville theory in the regime $c_\mathrm{L}\in\C\backslash(-\infty,1]$~\cite{caoAnalyticalContinuationLattice2023}.

These developments motivated the proposal of~\cite{usciatiProbabilisticConstructionNoncompactified2026}, which serves as the starting point of this article. In that work, the authors gave an explicit path-integral defition of correlation functions with central charge less than one, without restricting to charge neutrality. Following the ideas of~\cite{harlowAnalyticContinuationLiouville2011a}, they define the Imaginary Liouville measure by integrating over a complex contour. The choice of this contour is inspired both by the lattice analysis~\cite{caoAnalyticalContinuationLattice2023} and by the requirement that the resulting correlation functions reduce to the Coulomb gas integrals in the charge neutrality regime (see Proposition~\ref{defi:CG_correls}).

More precisely, let $\mc{F}_0$ be a set of generalized functions on the Riemann sphere $\phi_0: \hat{\C} \to \R$, with vanishing mean $\int_\C \phi_0(z) \dd v_g(z) = 0$ (where $g$ is a Riemannian metric on $\hat{\C}$ and $\dd v_g$ denotes the associated volume form). The Imaginary Liouville measure, defining the law of the Liouville field, reads
\begin{eqs}
\label{eq:formal_imliou}
    \ps{F} \coloneqq \int_{\mc{U}} \int_{\mc{F}_0} F[\bm{c} + \phi_0] e^{-S^g_{\mathrm{IL},\beta}(\bm{c}+\phi_0)} \; D\phi_0\; \dd \bm{c} \ .
\end{eqs}
By analogy with ordinary Liouville theory, $\bm{c}$ is referred to as the \textit{zero mode}. Unlike the standard case though, it is integrated over an infinite oriented contour in the complex plane:
\begin{align*}
    \mc{U} := \left(-\im \infty, 0\right] \cup \left[0,\tfrac{2\pi}{\beta}\right] \cup \left[\tfrac{2\pi}{\beta} , \tfrac{2\pi}{\beta} - \im \infty\right) \ .
\end{align*}
The action formally coincides with the standard Liouville action, but with a purely imaginary coupling constant $\gamma=\i\beta$:
\begin{eqs}
\label{eq:im_liou_action}
    &S^g_{\mathrm{IL},\beta}(\bm{c}+\phi_0) := \frac{1}{4\pi} \int_{\C} \bigg(|d_g \phi_0|^2 + \im Q K_g (\bm{c}+\phi_0(z)) + 4\pi \mu e^{\im \beta (\bm{c}+\phi_0(z))}\bigg) \dd v_{g}(z)
\end{eqs}
with $\im = \sqrt{-1}$, $0 < \beta < 2$ can be rational or irrational, $Q = \tfrac{\beta}{2} - \tfrac{2}{\beta}$, $\mu>0$. $d_g$ denotes the exterior derivative and $K_g$ the scalar curvature associated with the metric $g$. The probabilistic construction requires $\beta<\sqrt2$, assumption under which we work in the sequel.

In~\cite{usciatiProbabilisticConstructionNoncompactified2026}, the authors claimed that for any $\{x_1,\dots,x_N\} \subset \C^N$ and $(\alpha_1,\dots,\alpha_N)\in \R^N$, the average~\eqref{eq:formal_imliou} is well defined for the product of \textit{vertex operators}
\begin{align*}
    &F \coloneqq \prod_{j=1}^N V_{\alpha_j}(x_j) \ , \quad V_{\alpha_j}(x_j) [\bm{c}+\phi_0]\coloneqq e^{\im \alpha_j (\bm{c}+\phi_0(x_j))}
\end{align*}
provided that $\alpha_j > Q$ and $\sum_{j=1}^N \alpha_j - 2Q < 0$ (the so-called Seiberg bounds). In that range of parameters, they gave numerical evidence that
\begin{align*}
    \ps{V_{\alpha_1}(0)V_{\alpha_2}(1) V_{\alpha_3}(\infty)} = \ImDOZZ(\alpha_1,\alpha_2,\alpha_3) \ .
\end{align*}

\subsection{A probabilistic approach to Imaginary Liouville theory}
In this paper, we study the conformal symmetry of the path integral definition for Imaginary Liouville theory. To do so we rely on its probabilistic definition and consider its intrinsic features, beyond exact computations. Though we do not prove global conformal covariance, we show that the theory enjoys local conformal invariance via the derivation of local Ward identities. This is achieved assuming a conjecture, made in~\cite{usciatiProbabilisticConstructionNoncompactified2026}, ensuring that the zero-mode integral converges. Remarkably, we show that local conformal invariance is possible thanks to invariance of the zero-mode integral with respect to shifts in the \textit{imaginary} direction (Proposition~\ref{KPZ}) while global conformal invariance would follow by invariance by shifts in the \textit{real} direction (Remark~\ref{rmk}).

\subsubsection{Probabilistic construction via the path integral}
A key feature of the proposal~\eqref{eq:formal_imliou} is that the measure over $\phi_0$ admits a probabilistic formulation in terms of the GFF. More precisely, let $\X^g$ be the GFF on the Riemann sphere, viewed as a random generalized function on a probability space $(\Omega,\mc F,\mathbb{P})$. Then one can interpret
\begin{align*}
    \int_{\mc{F}_0} F[\phi_0] e^{-\frac{1}{4\pi} \int_{\C} |d_g\phi_0|^2 \dd v_g(z)} \equiv C(g) \ga{F[\X^g]} \ ,
\end{align*}
where $\ga{\cdot}$ denotes the expectation with respect to the GFF and $C(g)\coloneqq  \left(\frac{v_g(\hat{\C})}{\det^\prime\left( -\frac{\Delta_g}{2\pi}\right)}\right)^\frac{1}{2}$ is a regularized determinant. After a suitable regularization procedure, both the vertex operators and the exponential interaction term in~\eqref{eq:im_liou_action} can be defined as functionals of the GFF. The latter is interpreted as an instance of \textit{Imaginary Gaussian multiplicative chaos} (Imaginary GMC), denoted $M^g_\beta\d v_g$, and viewed as a random complex distribution. In fact, these probabilistic objects have already been defined in the \textit{compactified} setting~\cite{guillarmouCompactifiedImaginaryLiouville2025}, where $\beta^2$ is rational and both $Q$ and the vertex charges $\alpha_j$ are integer multiples of the inverse of the compactification radius $R^{-1}$. The correlation functions become, for $g=g_0$ the round metric on $\hat\C$ (see Subsection~\ref{subsec:GFF}),
\begin{eqs}\label{eq:correl_intro}
    \ps{\prod_{k=1}^NV_{\alpha_k}(x_k)} = C(g_0)e^{P_{\bm z,\bm\alpha}}&\int_{\mc U}e^{\i(\sum_{k=1}^N\alpha_k-2Q)\beta \bm{c}} \expect{\exp\left(-\mu e^{\i\beta\bm c}M_{\beta}\left(e^{\i\beta H_{\bm x,\bm\alpha}}\right)\right)}\dd \bm{c},\\
    \qt{where}M_\beta(e^{\i\beta H_{\bm x,\bm\alpha}})= &\int_\C e^{-\sum_{k=1}^N\beta\alpha_k G(y,x_k)}M^g_\beta\d v_g(y).
\end{eqs}

The principal difference between~\cite{guillarmouCompactifiedImaginaryLiouville2025} and the present work lies in the integration of the zero mode along the contour $\bm{c}\in \mc{U}$. Establishing the convergence of the integral over $\bm{c}$  in~\eqref{eq:correl_intro} remains an open problem. We formulate this problem in terms of the decay, as $\bm c\to -\i\infty$, of the expectation appearing in the definition of the correlation functions (Conjecture~\ref{conj:Laplace1}). In the following, we denote by $\mc A_N$ the set of weights $(\alpha_k)_{1\leq k\leq N}$ for which correlations are well-defined.

In the present work, we also provide new heuristic evidence in support of this conjecture (see Appendix~\ref{appendix:laplace}). This leads us to refine and extend the Seiberg bounds proposed in~\cite{usciatiProbabilisticConstructionNoncompactified2026}. Namely, we propose (Conjecture~\ref{conj:Laplace2}) that $\mc A_N$ is given by
\begin{align*}
    \alpha_k > Q\text{ for $1\leq k\leq n$} \quad ; \quad \sum_{k=1}^N \alpha_k - 2Q < \min_{k=1,\dots,N} 2(\alpha_k-Q) \wedge \frac{4}{\beta}\,\cdot
\end{align*}

\subsubsection{Ward identities and BPZ equations} The derivation of the structure constants in Imaginary (and standard) Liouville theory can be achieved thanks to the predicted \textit{local conformal invariance} of this two-dimensional QFT. 
A manifestation of this symmetry is the existence of a special functional of the Liouville field, the Stress-Energy Tensor (SET) $\SET(x)$,  which encodes the infinitesimal variations of the correlation functions with respect to the background metric. In the vertex algebra setting, the SET admits a formal series expansion whose modes satisfy the commutation relation of the Virasoro algebra.  
In the present work, we construct the SET of Imaginary Liouville theory in terms of the GFF $\X^g$. The Virasoro generators are defined indirectly by their action on correlation functions: they dress the vertex operator $V_\alpha$ by derivatives of the GFF. For instance $\L_{-1}V_{\alpha}$ is the functional formally defined by $\L_{-1}V_{\alpha}[\phi]=\alpha\partial\phi e^{\alpha\phi}$. Correlation functions with descendants fields (i.e. containing expressions of the form $\L_{-n}V_{\alpha_1}$) are defined in Definitions~\ref{eq:def_desc} and~\ref{def:der_cor}.
We then prove that for every $n\geq 2$, the correlation functions satisfy the \textit{conformal Ward identities} (see Theorem~\ref{thm:desc} for a precise statement): 
\begin{align*}
    \ps{\L_{-n}V_{\alpha_1}(x_1) \prod_{k=2}^{N} V_{\alpha_k}(x_k)} = \left(\sum_{k=2}^N -\frac{\partial_{x_k}}{(x_k-x_1)^{n-1}} + \frac{\Delta_{\alpha_k}}{(x_k-x_1)^{n}}\right) \ps{\prod_{k=1}^N V_{\alpha_k}(x_k)}
\end{align*}
where $\Delta_{\alpha_k}\in \R$ denotes the \textit{conformal dimension} of the vertex operator $V_{\alpha_k}$. Specializing to $\alpha_1=0$, we recover the Ward identity for the SET (Proposition~\ref{prop:ward_set}). Choosing $\alpha_1\in\{-\frac{\beta}{2},\frac{2}{\beta}\}$ so that $V_{\alpha_1}$ is a \textit{degenerate field at level two}, we further derive the celebrated \textit{BPZ equations} (see Theorem~\ref{thm:BPZ} for more details):
\begin{eqs}
\label{eq:intro_bpz}
    \left( - \frac{1}{{\alpha_1}^2} \partial_{x_1}^2 + \sum_{k=2}^N \left(\frac{\partial_{x_k}}{x_1-x_k} + \frac{\Delta_{\alpha_k}}{(x_1-x_k)^2}\right)\right) \ps{ \prod_{k=1}^N V_{\alpha_k}(x_k) } = 0 \ .
\end{eqs}
Our proof is close in spirit to the probabilistic approach developed for Liouville theory~\cite{kupiainenLocalConformalStructure2019}, and, more recently, for Toda CFTs~\cite{cercleWardIdentities$mathfraksl_3$2022}. We first regularize the Gaussian free field at a scale $\eps>0$, replacing the distribution-valued random variable by the smooth Gaussian field $\X^g_\eps$. We also regularize the exponential interaction by excluding neighborhoods of radius $\rho>0$ around the insertion points of the vertex operators. Combined with our estimates on the zero mode integral, these regularizations yield smooth correlation functions satisfying the Ward identities and BPZ equations up to explicit \textit{remainder terms} depending on $\eps$ and $\rho$.

Showing that the remainder terms go to zero as $\eps\to 0$ then $\rho\to 0$ is the main technical challenge of the paper. To achieve this, we establish \textit{fusion estimates} for Imaginary Liouville correlation functions. These estimates, known in the physics litterature as the Operator product expansion (OPE), constitute the most novel technical result of the paper. Indeed, their proof relies on the properties of the Imaginary Gaussian multiplicative chaos, whose behavior differs substantially from that of ordinary Gaussian multiplicative chaos arising in Liouville theory. We find the following discrete OPE (see Lemma~\ref{lemma:fusion_general} for a precise statement): let $m\in \N$ be such that $m-2<\beta(\alpha_1+\alpha_2) < m-1$, then one has
\begin{eqs}
\label{eq:intro_fusion}
&\ps{V_{\alpha_1}(0) V_{\alpha_2}(x) \prod_{k=3}^N V_{\alpha_k}(x_k)} = |x|^{\alpha_1\alpha_2} \left(\ps{V_{\alpha_1+\alpha_2}(0) \prod_{k=3}^N V_{\alpha_k}(x_k)} + \sum_{1\leq i+j \leq m} C_{i,j} x^i \overline{x}^j\right)\\
    &\qquad + |x|^{\alpha_1\alpha_2+\beta(\alpha_1+\alpha_2)+2} \rho(\beta\alpha_1,\beta\alpha_2) \ps{V_{\alpha_1+\alpha_2+\beta}(0)\prod_{k=3}^N V_{\alpha_k}(x_k)} + \underset{x\to 0}{o}(|x|^{\alpha_1\alpha_2+\beta(\alpha_1+\alpha_2)+2}) \ .
\end{eqs}
If we gather the operators appearing in the OPE by increasing conformal dimension, the integer $m$ thus corresponds to the number of Virasoro descendants separating the primaries $V_{\alpha_1+\alpha_2}$ and $V_{\alpha_1+\alpha_2+\beta}$.

\subsubsection{Towards conformal bootstrap for Imaginary Liouville theory} The results of the present work provide the first analytical evidence both in physics and in mathematics that the conformal bootstrap program may be carried out for CFTs with central charge smaller than one. It is the first work investigating the intrinsic features of the path integral construction of Imaginary Liouville theory in the realm of two-dimensional CFT. 

In particular, the BPZ equations provide a route towards an exact derivation of the structure constants. Assuming that the correlation functions satisfy global conformal covariance, the BPZ equations~\eqref{eq:intro_bpz} with $N=4$ map to hypergeometric differential equations. Combining their solutions with the fusion estimates~\eqref{eq:intro_fusion}, one can derive Teschner's shift relations for the structure constants as in Liouville theory~\cite{kupiainenIntegrabilityLiouvilleTheory2020}. Our results therefore provide strong evidence that the Imaginary DOZZ formula can be derived rigorously within the path-integral framework.

The fusion estimates~\eqref{eq:intro_fusion} also offer new insight into the spectrum of Imaginary Liouville theory. In particular, they appear in contradiction with the widely accepted idea that Imaginary Liouville theory possesses a purely continuous spectrum, as proposed in~\cite{ribaultLiouvilleTheoryCentral2015}. Instead, our results suggest that, as in ordinary Liouville theory, a complete bootstrap formula may require to take into account discrete contributions to the spectrum.

Finally, two major challenges remain. The first is to establish global conformal covariance (or \textit{Weyl invariance}) of the correlation functions. The second is to prove the convergence of the zero-mode integral. We hope to adress these questions in future works.

\textit{\textbf{Acknowledgements.}} The authors are grateful to Raoul Santachiara for his constant interest and support of this work. They also thank Colin Guillarmou and Rémi Rhodes for valuable discussions during the early stages of the project. R.~U. acknowledges the hospitality of the LPTHE.

\section{Definitions, preliminary results and statement of the Ward identities}
\subsection{Probabilistic background}
We introduce in this subsection the basic material that enters the definition of the correlation functions of Imaginary Liouville CFT.

\subsubsection{Gaussian free field}\label{subsec:GFF}
By means of stereographic projection, the sphere $\mathbb S^2$ equipped with its standard metric is identified with the Riemann sphere $\hat\C$ equipped with the Riemannian metric $g_{0}=e^{2w_0}\norm{dz}^{2}$ where $w_{0}(z)=-\log\left(1+\norm{z}^{2}\right)$. More generally we will consider metrics on $\hat\C$ of the form $g=e^{2w}\norm{dz}^{2}$ with $w$  (resp. $w\circ\theta+\ln\norm{\theta'}$ with $\theta:z\mapsto\frac1z$) smooth over $\C$ (resp. $\hat\C\setminus\{0\}$). We denote the set of such functions by $\mc C^\infty(\hat\C)$ and the set of corresponding metrics by $\Met\coloneqq \left\{g=e^{2w}\norm{dz}^{2},\, w\in C^\infty(\hat\C)\right\}$.
For such a metric we let $m_g(\cdot) \coloneqq v_g(\hat{\C})^{-1} \int_\C \cdot \; \dd v_g$ be the average of a $v_{g}$-integrable function, where the latter is the volume form in the metric $g$.

For such a $g$, let us denote for $x\neq y$ on $\C$ 
\begin{eqs}
\label{eq:green}
    \green(x,y) = \ln \frac{1}{\norm{x-y}} -\frac{1}{2}\left(W_{g}(x)+W_{g}(y)\right),\qt{with}\\
    W_{g}(x)\coloneqq 2m_{g}\left( \ln \frac{1}{\norm{x-\cdot}}\right)- \frac{1}{v_g(\hat{\C})^2} \int_{\C^2} \ln \frac{1}{|z-z^\prime|} \dd v_g(z) \dd v_g(z^\prime).\
\end{eqs}
If $g=g_0$, then $W_g=w_0+\theta_0$, $\theta_0=\frac12+\frac12\ln2$ (see~\cite[Equations (2.12) and (2.13)]{davidLiouvilleQuantumGravity2016a}). This kernel is such that for any $x$ in $\C$ we have $m_{g}(\green(x,\cdot))=0$. It is the Green function for the (non-positive) Laplacian $\Delta_{g}$ in that, in the distributional sense,
\begin{equation*}
   - \Delta_{g}\green(x,\cdot) = 2\pi\left(\delta_{x}-\frac{1}{v_g(\hat{\C})}\right).\
\end{equation*}
Alternatively, let $(\lambda_j)_{j\geq 1}$ be the positive eigenvalues of $-\Delta_g$, with corresponding eigenfunctions $(\varphi_j)_{j\geq 1}$. Then in the weak sense
\begin{equation*}
    G_g(x,y) = \sum_{j=1}^{+\infty} \frac{\varphi_j(x)\varphi_j(y)}{\lambda_j}\cdot
\end{equation*}
Now let $(a_j)_{j\geq 1}$ be independent Gaussian random variables with law $\mathcal{N}(0,1)$, defined on some probability space $(\Omega,\mathcal{F},\mathbb{P})$. Denote by $H^{-s}(\hat{\C},g)$ the Sobolev space of negative order $-s<0$ based on the $L^2$ space of square-integrable functions in the metric $g$. The Gaussian Free Field (GFF for short) is then the random element of $H^{-s}(\hat{\C},g)$ for any $s<0$ given by
\begin{equation}
    \X^g \coloneqq \sum_{j\geq 1} a_j \frac{\varphi_j}{\sqrt{\lambda_j}}\cdot
\end{equation}
The GFF is such that for any two test functions $u,v$,
\[
    \ga{\langle u, \X^g\rangle_{L^2} \langle v, \X^g \rangle_{L^2}} = 2\pi \int_{\C^2} u(x) \green(x,y) v(y) \; \dd v_g(x) \dd v_g(y) \ .
\]
By definition \eqref{eq:green} of the Green's function,  for $g^\prime$ and $g$ in $\Met$,
\begin{equation}\label{eq:GFF_met}
        \X^{g^\prime} \overset{\text{law}}{=} \X^g - m_{g^\prime}(\X^g) \ .
\end{equation}

Since the GFF is a distribution rather than a function, it is standard to regularize it. To do so, let $g\in\Met$, $x\in\hat\C$ and $\eps>0$. We define for $f$ integrable on the geodesic circles $\partial B(x,\eps)$ (of radius $\eps$ and centered at $x$) equipped with the length measure $l_g$,
\begin{eqs}\label{eq:def_reg_GFF}
   \mu_{g}(x,\eps)[f]\coloneqq\frac1{l_g\left(\partial B(x,\eps)\right)}\int_{\partial B(x,\eps)}f\;\dd l_{g}.
\end{eqs}
We set $\regset_g\coloneqq\left\{f\in H^{-1}(\hat\C,\C),\,x\mapsto \mu_{g}(x,\eps)[f]\in C^0(\hat\C,\C)\text{ for any }\eps>0 \right\}$. In particular, if $f=\X^{g}+H$ for some smooth $H$, then (up to a modification) $f\in\regset_g$ almost surely (see~\cite[Lemma 3.2]{guillarmouPolyakovsFormulation$2d$2019a}).
We note the variance formula (see~\cite[Equation (3.10)]{guillarmouPolyakovsFormulation$2d$2019a})
\begin{eqs}\label{eq:lim_var}
    \E[ \left(\mu_g(x,\eps)[\X^{g}]\right)^{2} ] =- \ln\eps - W_{g}(x)+o(1).
\end{eqs}
We also write $\X^g_\eps(x)\coloneqq \mu_{g,\eps}(x)[\X^g]$, and more generally $f_\eps(x)\coloneqq \mu_{g_0,\eps}(x)[f]$ for $f\in\regset_{g_0}$.

\subsubsection{Vertex operators and Gaussian multiplicative chaos} We now define two types of exponential functionals of the free field that will be instrumental in the following: vertex operators, which represent the primary fields of the conformal field theory, and the (imaginary) Gaussian multiplicative chaos (GMC), i.e. the exponential term in the Liouville action. 
\begin{defi}[Vertex operator]
For $\alpha \in \R$, $\eps>0$ and $g\in\Met$, define for $\phi\in\regset_g$
\begin{eqs}
&V_{\alpha,\eps}^g(x)[\phi] \coloneqq \eps^{-\frac{\alpha^2}{2}} e^{\i\alpha \mu_g(x,\eps)[\phi]}.
\end{eqs}
\end{defi}
By combining Equations~\eqref{eq:GFF_met} and~\eqref{eq:lim_var} we obtain for $g=e^{2\omega} g_0$
\begin{eqs}\label{eq:cov_VO}
    V^{g}_{\alpha,\eps}(x)[\X^{g}-\i Q \omega+\bm c] \eqlaw e^{-2\Delta_\alpha \omega(x)} V^{g_0}_{\alpha,\eps}(x)[\X^{g_0}+\bm c- m_{g}(\X^{g_0})](1+\underset{\eps\to0}{o}(1)) \ .
\end{eqs}
Here $\Delta_\alpha = \frac{\alpha}{2}\left(\frac{\alpha}{2}-Q\right)$ is the \emph{conformal dimension} of the vertex operator $V_{\alpha}$.
\begin{defi}[Imaginary GMC] Let $\beta \in [0, \sqrt{2})$ and $g=e^{2w}\norm{dz}^2\in\Met$. Then
\begin{equation}
M_{\beta,\eps}^g\dd v_g \coloneqq  V^{g}_{\beta,\eps}(x)[\X^{g}-\i Q w]\dd^2 x
\end{equation}
converges (\cite[Theorem 3.1]{lacoinComplexGaussianMultiplicative2015b}), as $\eps\to0$, in probability and in $H^{s}(\hat{\C})$ for $s<-1$.
\end{defi}
We denote by $M_{\beta}^g(\dd^2 x) $ this limit.
By Equation~\eqref{eq:cov_VO}, since $\Delta_\beta=1$, for $g=e^{2\omega}g_0$ in $\Met$:
\begin{equation}\label{eq:cov_GMC}
e^{\i\beta\bm c}M_{\beta}^g\dd v_g\eqlaw e^{\i\beta\left(\bm c-m_g(\X^{g_0})\right)}M_{\beta}^{g_0}\dd v_{g_0} .
\end{equation}
Finally, if $U\subset \C$ Borel and $f$ is a test function, we put
\begin{equation}
    M_{\beta}^g(U,f) \coloneqq \int_U f(x) M_{\beta}^g\dd v_g.
\end{equation}
Further we set $M_\beta^g(U) = M_\beta^g(U,1)$ for $U$ Borel, and $M_{\beta}^{g}(f)=M_{\beta}^g(\C,f)$ for $f$ as above.

\subsubsection{On the Laplace's transform of imaginary GMC}
Correlation functions in Imaginary Liouville theory are expressed using the Laplace's transform of such imaginary GMC integrals, i.e. expressions of the form
$\expect{e^{-\mu M^g_\beta(f)}}$, for some $\mu\in\C$ and $f$ a measurable function. Typically, $f$ does not necessarily have compact support and is singular in the neighborhood of distinct points $x_{1},\cdots,x_{N}$ in $\C$. To ensure existence of such quantities we rely on:
\begin{lemma}[Exponential moments]\label{lemma:def_laplace}
Let $f$ be of the form $f(x)=\prod_{k=1}^{N}\norm{x-x_{k}}^{\beta\alpha_{k}}e^{2h(x)}$, where $h$ is continuous on $\C$ and such that $f$ is bounded at $\infty$. Further assume that $\alpha_{k}>Q$ for all $1\leq k\leq N$. Then
\begin{equation}
\expect{ \norm{e^{-M^g_\beta(f)}}}<\infty.
\end{equation}
\end{lemma}
\begin{proof}
This statement directly comes from the proof of the finiteness of exponential moments~\cite[Prop.~5.3 and 6.4]{guillarmouCompactifiedImaginaryLiouville2025}, where the authors obtain the bound
\begin{align*}
    &\ga{e^{\norm{M^g_\beta(f)}}} \leq e^{C_1 v(A) + C_2v(A^c)} \left(1+C_1u(A)e^{C_1u^2(A)}\right) \left(1+C_2 u(A^c) e^{C_2u^2(A^c) }\right) \\
    &\qquad u^2(A) \coloneqq \int_{A^2} \norm{f(x)} \norm{f(y)} l_g(x,y)^{-\beta^2} \dd v_g(x) \dd v_g(y)\\
    &\qquad v(A) \coloneqq \int_{A} \norm{f(x)} \text{dist}(x,\partial A)^{-\frac{\beta^2}{2}} \dd v_g(x)
\end{align*}
where we choose $A=\D$ the unit disk, and the constants $C_1$ and $C_2$ do not depend on $f$.

The integrals $u^2$ are finite since the only possible divergences would occur in the vicinity of some $x_k$, where they have the same behavior as
\begin{align*}
    \int_{B(0,1)^2} |x|^{\beta\alpha_k} |y|^{\beta\alpha_k} |x-y|^{-\beta^2} \dd^2 x \dd^2 y  &= 2\int_{0\leq |x|\leq |y| \leq 1} |x|^{\beta\alpha_k} |y|^{\beta\alpha_k} |x-y|^{-\beta^2} \dd^2 x \dd^2 y\\
    &= \int_{B(0,1)} |y|^{2\beta \alpha_k + 2 - \beta^2 } \left( \int_{B(0,1)}  |x^\prime|^{\beta\alpha_k} |x^\prime-1|^{-\beta^2} \dd^2 x^\prime \right)\dd^2 y
\end{align*}
which are indeed finite since $\beta^2 < 2$ and $\alpha_k > Q$. The same arguments show that $v(\D)$ and $v(\D^c)$ are finite as well.
\end{proof}

We now state our main conjecture about the asymptotic behavior of the Laplace transform $\expect{e^{-\mu M^g_\beta(f)}}$ as $\mu\to+\infty$. This conjecture is the key point that is currently missing in our analysis. A discussion on the motivations for such a conjecture is given in Appendix~\ref{appendix:laplace}: it arises from both numerical and analytical evidence.

\begin{conjecture}
\label{conj:Laplace1}
Let $f$ be as in Lemma~\ref{lemma:def_laplace}, and assume that $\alpha_k > Q$ for all $1\leq k\leq N$. There exists $\lambda=\lambda_\beta(\bm\alpha)>0$ such that:
let $\zeta >0$, $K\subset\C^{N-1}$ be compact and
\begin{equation*}
    K_\infty \subset\left\{\bm x=(x_{1},\cdots,x_{N})\in K \times \C, \; \norm{x_i-x_j}\geq \zeta\text{ for }i\neq j\right\} \ .
\end{equation*}
Then for any $f$ as in Lemma~\ref{lemma:def_laplace} and $A>0$,
\begin{center}
\fbox{$\sup\limits_{\mu\geq1}\sup\limits_{\bm{x}\in K_\infty}\sup\limits_{\nnorm{u}_{\infty}\leq A} \mu^{\lambda} \norm{\expect{e^{-\mu M^g_\beta(fe^{u})}}} < +\infty$} \ .
\end{center}
\end{conjecture}

\begin{conjecture}\label{conj:Laplace2}
    We conjecture that $\lambda=\lambda_\beta(\bm\alpha)$ is given by
\begin{equation}
\lambda\coloneqq \min\limits_{j=1,\dots,N} \frac{2(\alpha_j -Q)}{\beta} \wedge \frac{4}{\beta^2}\cdot
\end{equation}
\end{conjecture}

\subsection{Definition of the correlation functions}
Based on this material we can now provide the definition of the correlation functions for the imaginary Liouville theory. To this end let us pick distinct insertions $x_{1},\cdots,x_{N}$ in $\C$ and weights $\alpha_1,\cdots,\alpha_{N}$ in $(Q,+\infty)$. In the rest when $g=g_0$ we omit the supercript $g$: for instance $\X\coloneqq \X^{g_0}$. 

\subsubsection{Correlation functions of Vertex Operators}
For positive $R$ and $\eps$, we define the (regularized) correlation functions by setting 
\begin{eqs}\label{eq:correl_reg}
    \ps{\prod_{k=1}^N V_{\alpha_k}(x_k)}_{g,R,\eps}&\coloneqq C(g)\int_{\gamma_0([-R,R])}(\gamma_0)_\ast(\dd \bm{c})\\
    &  \ga{\prod_{k=1}^N V_{\alpha_k,\eps}^g(x_k)[\X^g+\i Q w_0+\bm c] e^{-\frac{\i Q}{4\pi} \int_{\C} K^g (\X^{g}+\bm c) \dd v_g- \mu \int_\C V^g_{\beta,\eps}(y)[\X^g+\i Q w_0+\bm c]\dd^2 y} } \; 
\end{eqs}
where $\gamma_0: \R\to\C$, described below, is a continuous injective map defining an infinite integration contour for the zero mode $\bm{c}$ in the complex plane. To first take the $\eps\to0$ limit of the above, we rewrite the expectation value using Girsanov's theorem.
By combining Equation~\eqref{eq:GFF_met} and $K^{g}\dd v_g=(2-2\Delta_{g_0}\omega)\dd v_{g_0}$ for $g=e^{2\omega}g_0\in\Met$, we get
\begin{align*}
    e^{-\frac{\i Q}{4\pi} \int_{\C} K^g(\X^{g}+\bm c) \dd v_g}\eqlaw e^{-2\i Q\left(\bm c-m_{g_0}(\X^g)\right)+\frac{\i Q}{2\pi}\int_\C\X^g\Delta_{g_0}\omega\dd v_{g_0}}.
\end{align*}
We now apply Girsanov's theorem, but with a shift of the law of $\X$ in the \textit{imaginary} direction, namely by $-\i Q(w-m_{g_0}(w))$. The result still holds but one has to go through an analytic continuation argument, which has been already discussed in~\cite[Th.~6.11(1)]{guillarmouCompactifiedImaginaryLiouville2025}.
The expectation term in Equation~\eqref{eq:correl_reg} becomes, up to a $o_\eps(1)$ and with $\bm c'=\bm c-m_{g_0}(\X^g)$,
\begin{align*}
    &e^{-\frac{Q^2}{4\pi}\int_\C \left(-w\Delta_{g_0}\omega+4\omega\right)\dd v_{g_0}}\prod_{k=1}^N e^{-2\Delta_{\alpha_k}\omega(x_k)}\ga{e^{-2\i Q\bm c'}\prod_{k=1}^N V_{\alpha_k,\eps}(x_k)[\X+\i Q w_0+\bm c'] e^{- \mu e^{\i\beta\bm c'}M^g_\beta(\C)}}.
\end{align*}
where we have used Equations~\eqref{eq:cov_VO} and~\eqref{eq:cov_GMC}. The same strategy applied to the Vertex Operators gives rise to the following terms: set
\begin{eqs}\label{eq:def_HP}
 	H_{\bm x,\bm\alpha}\coloneqq \i\sum_{k=1}^{N}\alpha_{k}G(\cdot,x_{k}),\quad P_{\bm x,\bm\alpha}\coloneqq -\sum_{k< l}\alpha_{k}\alpha_{l}G(x_{k},x_{l})+\sum_{k=1}^N\left(\Delta_{\alpha_k} w(x_k)+\frac{\alpha_k^2}{2}\theta_0\right).
\end{eqs}
For $\bm c\in\C$, further define
\begin{equation}\label{def:Ex}
	\Ex_{\bm x,\bm\alpha}(\bm c)\coloneqq e^{\i\left(\sum_{k=1}^{N}\alpha_k-2Q\right)\bm c}\expect{\exp\left(-\mu e^{\i\beta\bm c}M_{\beta}\left(e^{\i\beta H_{\bm x,\bm\alpha}}\right)\right)}.
\end{equation}
In view of Lemma~\ref{lemma:def_laplace} the latter is well-defined since $\alpha_{k}>Q$ for all $1\leq k\leq N$, and

\begin{eqs}\label{eq:correl_reg2}
	&\ps{\prod_{k=1}^N V_{\alpha_k}(x_k)}_{g,R} \coloneqq\lim\limits_{\eps\to0}\ps{\prod_{k=1}^N V_{\alpha_k}(x_k)}_{g,R,\eps} \\
    &=\prod_{k=1}^N e^{-2\Delta_{\alpha_k}\omega(x_k)}e^{\frac{c_{\mathrm{IL}}}{96\pi}S_0(2\omega)}C(g)e^{P_{\bm z,\bm\alpha}} \int_{\gamma_0([-R,R])}\Ex_{\bm z,\bm\alpha}\left(\bm c-m_{g_0}(\X^{g})\right)) \; (\gamma_0)_\ast(\dd \bm{c})
\end{eqs}
with a slight abuse of notation. Here $c_{\mathrm{IL}} \coloneqq 1-6Q^2$ is the central charge and $S_0(\omega)=\int_\C\left(-\omega\Delta_{g_0}\omega+4\omega\right)\d v_{g_0}$.
\begin{remark}\label{rmk}
    We almost retrieve the Weyl anomaly which is an axiom of conformal field theory~\cite[Section~4.2]{guillarmouPolyakovsFormulation$2d$2019a}. A sufficient condition to obtain it would be to choose the zero-mode contour to be the real line, because the Lebesgue measure is translation-invariant. Then one could shift the zero-mode $\bm{c}\in \R$ by the Gaussian random variable $m_g(\X^{g_0})$ (this is what is usually done in ordinary Liouville theory~\cite[Lemma~3.1]{guillarmouPolyakovsFormulation$2d$2019a}). As we will see shortly, in our case the contour $\gamma_0([-R,R])$ is not translation-invariant in the real direction. However, since $m_{g_0}(\X^{g})$ only shifts the mean of $\X^g$, we believe that it does not change the asymptotic behavior predicted in Conjectures~\ref{conj:Laplace1} and~\ref{conj:Laplace2}. As such, global conformal covariance is equivalent to invariance of the integral over the zero mode by shifts in the \textit{real} direction.
\end{remark}

Let us now describe the contour of integration for the $\bm c$ variable. As explained in the introduction, it has been argued that Imaginary Liouville theory corresponds to a symmetric \textit{U-shaped contour} connecting the $-\i\infty$ and $(2\pi-\i\infty)/\beta$ regions of the complex plane and containing the interval $[0,2\pi/\beta]$:
\begin{equation}
    \gamma_0:\R\to\C, \quad x \mapsto \left\{\begin{array}{ll}
        \i(x+\frac{\pi}{\beta}) & \text{ if } x\leq-\frac\pi\beta,\\
        x+\frac\pi\beta & \text{ if } -\frac\pi\beta\leq x\leq\frac\pi\beta,\\
        \frac{2\pi}{\beta} + \i(x-\frac{\pi}{\beta}) & \text{ if } x \leq \frac\pi\beta.
    \end{array}\right.
\end{equation}
Because of the existence of exponential moments (Lemma~\ref{lemma:def_laplace}), $\Ex_{\bm z,\bm\alpha}(\bm c)$ as defined in Equation~\eqref{def:Ex} is a holomorphic function of $\bm{c}$. Hence by Cauchy theorem, the exact shape of the integration contour $\gamma_0(\R)$ does not matter as long as it is non self-intersecting and keeps the same asymptotic behavior $\gamma_0(-\infty) = -\i\infty$, $\gamma_0(+\infty) = 2\pi/\beta -\i\infty$. We write $\mathcal{U} \coloneqq \gamma_0(\R)$.
\begin{defi}[Correlation functions of vertex operators] \label{defi:correls} Let $(x_k)_{1\leq k\leq N}$ in $\C$ be distinct. We define $\mc A_N\coloneqq\left\{(\alpha_k)_{1\leq k\leq N}\in(Q,+\infty)^N,\, \sum_{k=1}^N\alpha_k-2Q<\lambda\right\}$ the set of charges satisfying the \textit{Seiberg bounds}. 
Then, assuming that Conjecture~\ref{conj:Laplace1} holds true, the following limit exists and is finite:
\begin{equation}
    \ps{\prod_{k=1}^N V_{\alpha_k}(x_k)}\coloneqq\lim\limits_{R\to+\infty}\ps{\prod_{k=1}^N V_{\alpha_k}(x_k)}_{R} \end{equation}
Such limits define the correlation functions of imaginary Liouville theory.
\end{defi}
For $g=g_{0}$, the contour integral can be rewritten more explicitly as follows:
\begin{eqs}
\label{eq:decomp_contour}
	\ps{\prod_{k=1}^N V_{\alpha_k}(x_k)} &= \im \left(1-e^{-2\i\pi s }\right)e^{P_{\bm z,\bm\alpha}}\int_{0}^{+\infty} e^{-s\beta \bm{c}} \expect{\exp\left(-\mu e^{\beta\bm c}M_{\beta}\left(e^{\i\beta H_{\bm x,\bm\alpha}}\right)\right)}\dd \bm{c}\\
	&\quad + e^{P_{\bm z,\bm\alpha}} \int_{0}^{\frac{2\pi}{\beta}} e^{-\im s\beta \bm{c}} \ga{\exp\left(-\mu e^{\im \beta \bm{c} }M_{\beta}\left(e^{\i\beta H_{\bm x,\bm\alpha}}\right)\right)} \dd \bm{c}
\end{eqs}
with $s\coloneqq(2Q-\sum_{k=1}^N \alpha_k)\big/\beta$. When $s$ is an integer, the first term on the right hand side vanishes while the second term is the Fourier transform of a periodic function. In particular, when $s$ is a negative integer the correlation functions identically vanish. Regarding the three-point function, this is consistent with the Imaginary DOZZ formula (see~\cite[Eq.~(7.10)]{guillarmouCompactifiedImaginaryLiouville2025}). Likewise when \textit{charge neutrality} holds, that is $s\in\N$, Conjecture~\ref{conj:Laplace1} is no longer needed to define the correlation functions:
\begin{proposition}[Coulomb gas correlation functions]
\label{defi:CG_correls}
In the setting of Definition~\ref{defi:correls}, take $g=g_{0}$ and further assume that $s$ is a non-negative integer. Then the correlation functions are well-defined and are given by Dotsenko-Fateev type integrals:
\begin{eqs}
\label{eq:charge_neutr}
    \ps{\prod_{k=1}^N V_{\alpha_k}(x_k)} = \; &\frac{2\pi}{\beta } \frac{(-\mu)^s}{s!} \prod_{k<l} |x_k-x_l|^{\alpha_j\alpha_k}  \int_{\C^s} \prod_{j=1}^s \prod_{k=1}^N |x_k-y_j|^{\beta \alpha_k} \prod_{j<l} |y_j-y_l|^{\beta^2} \dd^2 y_j. \ 
\end{eqs}
\end{proposition}

\subsubsection{Further functionals and derivatives of the Liouville field}
To derive Ward identities for the correlation functions, we need to define derivatives of the Liouville field. Let us assume that $F:C^\infty(\C,\C)\to C^\infty(\C,\C)$ is of the form $F:\Phi\mapsto\prod_{l=1}^{p}\partial^{n_{l}}\Phi(z)$ for some positive integers $n_{1},\cdots,n_{p}$ and $z\in\C$.  For such a $F$, we set for positive $R$, $\eps$
\begin{eqs}
    &\ps{F \prod_{k=1}^N V_{\alpha_k}(x_k)}_{R,\eps} \\
    &\coloneqq C(g_0) \int_{\gamma_0([-R,R])} e^{-2\i Q\bm c}\ga{ F[\Phi_\eps] \prod_{k=1}^N V_{\alpha_k,\eps}(x_k) e^
    {- \mu e^{\i\beta \bm{c}} M_{\beta,\eps}(\C)} } \; (\gamma_0)_\ast(\dd \bm{c}).
\end{eqs}
where $\Phi\coloneqq\X+\i Qw_0+\bm c$ is the \textit{Liouville field}, and $\Phi_\eps\coloneqq\X_\eps+\i Qw_0+\bm c$.
First assume that $z\in\C\setminus\{z_1,\cdots,z_N\}$: then by Girsanov's theorem we can rewrite the latter, up to a $o(1)$ term, as
\begin{align*}
    &C(g_0)e^{P_{\bm z,\bm\alpha}}\int_{\gamma_0([-R,R])} \ga{ F[\Phi_\eps+ H_{\bm x,\bm\alpha}] e^
    { - \mu e^{\i\beta \bm{c}} M_{\beta,\eps}(e^{\i\beta H_{\bm x,\bm\alpha}})} } \; (\gamma_0)_\ast(\dd \bm{c}).
\end{align*}
To make sense of this expression we then rely on Gaussian integration by parts:
\begin{lemma}[Gaussian integration by parts]
	\label{gaussian_integration}
	Let $X$ and $Y_1,\dots,Y_N$ be Gaussian random variables on the same probability space, and $f\in \mathcal{C}^\infty(\R^N,\C)$ with compact support. Then
	\begin{equation}
		\E\left[ X f(Y_1,\dots,Y_N) \right] = \sum_k \E[XY_k] \;  \E[\partial_k f( Y_1, \dots , Y_N )] \ .
	\end{equation}
\end{lemma} 
In our setting this yields the equality, for $\mu\in\C$ and $f$ continuous:
\begin{align*}
    &\ga{ \prod_{l=1}^{p}\partial^{n_{l}}\X_{\eps}(z) e^{ - \mu M_{\beta,\eps}(f)}}=\sum_{l=1}^{p-1} \partial^{n_{p}}_{z_1}\partial^{n_{l}}_{z_2}G_{\eps}(z_1,z_2)\vert_{z_1=z_2=z} \ga{ \prod_{m\neq l,p}\partial^{n_{m}}\X_{\eps}(z)e^{ - \mu  M_{\beta,\eps}(f)} }\\
    &-\mu \i\beta\int_{\C} \partial^{n_{l}}_zG_{\eps}(z,y)f(y)\ga{ \prod_ {l=1}^{p-1}\partial^{n_{m}}\X_{\eps}(z)V_{\beta,\eps}(y)e^{ - \mu  M_{\beta,\eps}(f)} }\norm{\dd y}^2.
\end{align*}
The first term is singular as $\eps\to0$. Likewise the second integral might become singular at $y=z$. To remedy this issue, we first introduce an extra regularization parameter $\rho>0$, and set $\C_\rho\coloneqq \C\cup\bigcup_{k=1}^N \setminus B(z_k,\rho) \cap B(0,\rho^{-1})$ for some $z_k$ fixed. We then define the $\eps\to0$ of the above recursively by means of Wick products by setting for $H$ smooth at $z$:
\begin{eqs}\label{eq:def_deri}
    :\prod_{l=1}^{p}\partial^{n_{l}}\left(\X_\eps+H\right)(z):\quad&\coloneqq\quad \partial^{n_{p}}\left(\X_\eps+H\right)(z)H(z):\prod_{l=1}^{p-1}\partial^{n_{l}}\left(\X_\eps+H\right)(z):\\
    &+\sum_{l=1}^{p-1}\partial^{n_p}_{z_1}\partial^{n_l}_{z_2}\ln\norm{z_1-z_2}\vert_{z_1=z_2=z} \prod_{m\neq l,p}\partial^{n_{m}}\left(\X_\eps+H\right)(z).
\end{eqs}
\begin{defi}[Derivatives of the Liouville field]
    Derivatives of the Liouville field are defined recursively by setting, for $H$ smooth at $z$, $\mu\in\C$, $f$ continuous and $\rho>0$:
    \begin{align*}
    &\ga{ :\prod_{l=1}^{p}\partial^{n_{l}}\left(\X+H\right)(z): e^{ - \mu M_{\beta}(\C_\rho,f)}}\coloneqq \lim\limits_{\eps\to0}\ga{ :\prod_{l=1}^{p}\partial^{n_{l}}\left(\X_\eps+H\right)(z): e^{ - \mu M_{\beta}(\C_\rho,f)}}.
\end{align*}
\end{defi}
Typically we will apply this definition to $H=H_{\bm x,\bm\alpha}+\im Q w_0$. However when $z=z_l$ for some $1\leq l\leq N$, the function $H_{\bm x,\bm\alpha}$ becomes singular. In that case we set 
\begin{equation}
    H_{\bm x,\bm\alpha}^{(l)}\coloneqq \i\sum_{k=1}^{N}\alpha_{k}\left(G(\cdot,x_{k})+\ln\norm{\cdot-z_k}\mathds 1_{k=l}\right)
\end{equation}
which is regular at $z_l$. Our definition of regularized correlation functions with derivatives is 
\begin{defi}[Correlation functions with derivatives of the field]\label{def:der_cor}
In the setting of Definition~\ref{defi:correls}, take $F:\Phi\mapsto \prod_{l=1}^p\partial^{n_l}\Phi$. We define
\begin{eqs}\label{eq:correl_der}
   &\ps{F(x_1) \prod_{k=1}^N V_{\alpha_k}(x_k)}_{R,\rho}\coloneqq C(g_0) e^{P_{\bm z,\bm\alpha}}\times \\
    &\int_{\gamma_0([-R,R])} e^{-2\i Q\bm c}\ga{ :F[\liou+H_{\bm x,\bm\alpha}^{(1)}](z_1): e^{ - \mu e^{i\beta \bm{c}} M^g_{\beta}\left(\C_\rho,e^{\im \beta H_{\bm x,\bm\alpha}}\right)} } \; (\gamma_0)_\ast(\dd \bm{c}).
\end{eqs}
\end{defi}
This quantity is indeed well-defined for $\rho>0$ since we integrate away from singular points. In the next section we will study the $\rho\to0$ limit\footnote{Convergence is ensured from Lemma~\ref{lemma:def_laplace} by uniform integrability.} of expressions of this form and show that they give rise to Ward identities satisfied by the correlation functions. Specializing to the $g=g_0$ case we are led to studying the $\rho\to0$ behaviour of
\begin{align*}
    &e^{P_{\bm z,\bm\alpha}}\int_{\gamma_0([-R,R])} \ga{ :\prod_{l=1}^p\partial^{n_l}[\liou+ H_{\bm x,\bm\alpha}^{(1)}](z_1): e^{- \mu e^{\i\beta \bm{c}} M^g_{\beta}(\C_\rho,e^{\i\beta H_{\bm x,\bm\alpha}})}} \; (\gamma_0)_\ast(\dd \bm{c}).
\end{align*}

\subsection{Statement of the main results}
The stress-energy tensor and the descendants fields are viewed as the following functionals $C^\infty(\C,\C)\to C^\infty(\C,\C)$ applied to the Liouville field:
\begin{defi}[Stress-energy tensor]
    The stress-energy tensor (SET) is the functional
	\begin{equation}
		\SET[\Phi] \coloneqq \im Q \partial^2 \Phi - ( \partial \Phi) ^ 2.
	\end{equation}
\end{defi}
\begin{defi}[Descendant fields] Let $\alpha\in\C$, $n\geq 1$. We define the functional
\begin{eqs}\label{eq:def_desc}
    \L_{-n}^\alpha[\Phi] = \im( (n-1)Q + \alpha ) \frac{\partial^n \Phi}{(n-1)!} - \sum_{i=0}^{n-2} \frac{\partial^{i+1} \Phi}{i!} \frac{\partial^{n-i-1} \Phi}{(n-2-i)!}\cdot
\end{eqs}
The descendant field $\L_{-n}V_\alpha$ is defined as $\L_{-n}V_\alpha[\Phi]\coloneqq \L_{-n}^\alpha[\Phi]  V_{\alpha}[\Phi]$.
\end{defi}
Insertion of such quantities within correlation functions is defined by~\ref{def:der_cor}. We then have:
\begin{theorem}[Ward identities for the descendant fields]
\label{thm:desc}
    Let $(x_k)_{1\leq k\leq N}\in\C^N$ be distinct and let $\alpha_1,\cdots,\alpha_N>-\frac{3}{2\beta}-\frac\beta2$. In either case:
    \begin{itemize}
        \item charge neutrality holds, \emph{i.e.}, $\frac{ 2Q- \sum_{k=1}^N \alpha_k}{\beta}= s\in\N$;
        \item $(\alpha_1,\cdots,\alpha_N,\underbrace{\beta,\cdots,\beta}_{k\text{ terms}})\in\mc A_{N+k}$ for $0\leq k\leq 2$ 
        and Conjecture~\ref{conj:Laplace1} is true;
    \end{itemize}
    then for any $n\geq 2$ and for $\alpha_1>\frac{n-4}{2\beta}-\frac\beta2$, in the sense of weak derivatives:
    \begin{eqs}
    \label{eq:thm_ward}
        \ps{ \L^{\alpha_1}_{-n} V_{\alpha_1}(x_1) \prod_{k=2}^N V_{\alpha_k}(x_k) } = \sum_{k=2}^N \left(-\frac{\partial_{x_k}}{(x_k-x_1)^{n-1}} + \frac{\Delta_{\alpha_k}}{(x_k-x_1)^n} \right) \ps{\prod_{k=1}^N V_{\alpha_k}(x_k) }.
    \end{eqs}
\end{theorem}
As corollaries, we deduce the two statements:
\begin{proposition}\label{prop:ward_set} In the setting of Theorem~\ref{thm:desc} with $n=2$,
\begin{eqs}
\label{eq:formal_global_Ward}
    \ps{ \SET(x) \prod_{k=1}^N V_{\alpha_k}(x_k) } = \sum_{k=1}^N \left(\frac{\partial_{x_k}}{x-x_k} + \frac{\Delta_{\alpha_k}}{(x-x_k)^2}\right) \ps{\prod_{k=1}^N V_{\alpha_k}(x_k) } \ .
\end{eqs}
\end{proposition}

\begin{theorem}[BPZ equation]\label{thm:BPZ} Assume that $\alpha_1=-\frac\beta2$ or $\alpha_1=\frac2\beta$. Then, in the setting of Theorem~\ref{thm:desc} with $n=2$,
\begin{eqs}
\label{eq:formal_BPZ}
    \left( - \frac{1}{{\alpha_1}^2} \partial_{x_1}^2 + \sum_{k=1}^N \left(\frac{\partial_{x_k}}{x_1-x_k} + \frac{\Delta_{\alpha_k}}{(x_1-x_k)^2}\right)\right) \ps{ \prod_{k=1}^N V_{\alpha_k}(x_k) } = 0 \ .
\end{eqs}
\end{theorem}

\section{Proofs of the main statements}
In this section, and unless specified, we assume that $g=g_0$ and drop the superscript $g$. We also assume that Conjecture~\ref{conj:Laplace1} holds.

\subsection{Regularity of correlation functions}
\subsubsection{Behavior of the correlation functions at infinity}
To prepare for the derivation of Ward identities, we need estimates on the behavior of (non-regularized) correlation functions seen as functions of the insertion points. We start by proving a bound controlling the behavior of correlation functions at infinity.

\begin{lemma}\label{lemma:estimate_infinity}
Let $(x_k)_{1\leq k\leq N}\in\C^N$ be distinct and $(\alpha_1,\cdots,\alpha_N,\beta,\beta)\in\mc A_{N+2}$. Let $\zeta >0$. Then uniformly on $x,x^\prime \in \C$ such that $\zeta < \left(d_{g_0} (x,x^\prime) \wedge \min\limits_{k=1,\dots,N} d_{g_0}(x,x_j) \wedge \min\limits_{k=1,\dots,N} d_{g_0}(x^\prime,x_j)\right)$, there exists $C \geq 0$ such that,
for both regularized and non-regularized correlation functions,
\begin{equation}
\label{eq:asymp_estimates}
\left|\ps{ V_\beta(x) V_\beta(x^\prime) \prod_{k=1}^N V_{\alpha_k}(x_k) } \right| \leq C(1+|x|)^{-4} (1+|x^\prime|)^{-4}.
\end{equation}
\end{lemma}
\begin{proof} 
For simplictity, we derive the result with a single vertex operator $V_\beta(x)$ in the correlation function since the proof is readily generalized in the case of two.
Set $\mathbf{V}\coloneqq \prod_{k=1}^N V_{\alpha_k}(x_k)$.
When charge neutrality holds, Eq.~\eqref{eq:correl_reg2} yields for $R,\eps >0$ (where $P_{\bm z,\bm\alpha,\eps}$ and $H_{\bm{z},\bm{\alpha},\eps}$ are expressed in terms of regularized Green functions):
\[
    \ps{ V_\beta(x) \mathbf{V} }_{R,\eps} = 2^{-\frac{\beta^2}{4}-\sum_{k=1}^N \frac{\alpha_k^2}{4}} e^{-\beta \sum_{k=1}^N \alpha_k G_\eps(x,x_k) + 2w_0(x)} e^{P_{\bm z,\bm\alpha,\eps}} \ga{\left(M_\beta(e^{-\beta^2 G_\eps(x,\cdot)} e^{i\beta H_{\bm{z},\bm{\alpha},\eps}})\right)^{s}}
\]
where $s=\frac{2Q - \beta - \sum_{k=1}^N \alpha_k}{\beta}\in\N$, and we have used that $\Delta_\beta=1$. This is actually independent of the cutoff $R$ on the zero mode, so we drop the subscript when charge neutrality holds. Since $\inf\limits_{x,y\in\C} G(x,y) - G_\eps(x,y)>c$ uniformly on $\eps$,  for some positive $A$
\[
    \sup\limits_{\eps>0} \norm{\ps{ V_\beta(x) \mathbf{V} }}_\eps \leq A e^{-\beta \sum_{k=1}^N \alpha_k G(x,x_k) + 2 w_0(x)} e^{P_{\bm z,\bm\alpha}} \ga{(M_\beta(e^{-\beta^2 G(x,\cdot)} e^{i\beta H_{\bm{z},\bm{\alpha}}}))^{s}}.
\]
Since  the Green function is bounded from below, i.e. $\|e^{-\beta^2 G(x,\cdot)}\|_\infty < B$ for some positive $B$,
\begin{align*}
    \sup\limits_{\eps>0} \norm{\ps{ V_\beta(x) \mathbf{V} }}_\eps &\leq A e^{-\beta \sum_{k=1}^N \alpha_k G(x,x_k) +2 w_0(x)} e^{P_{\bm z,\bm\alpha}} \\
    &\quad\times\int_{\C^{s}} \prod_{i=1}^{s} e^{-\beta^2 G(x,y_i)} e^{\i\beta H_{\bm{z},\bm{\alpha}}(y_i)} \prod_{i<j} e^{-\beta^2 G(y_i,y_j)} \dd v_g(y_i) \\
    &\leq AB e^{-\beta \sum_{j=1}^N \alpha_j G(x,x_j) - 2 w_0(x)} e^{P_{\bm z,\bm\alpha}} \ga{(M_\beta( e^{i\beta H_{\bm{z},\bm{\alpha}}}))^{s}} \ .
\end{align*}
We now bound the prefactors. The Green function is smooth outside the diagonal and explicit calculations show that $G(x,y) \underset{|x|\to \infty}{=} -\frac{1}{2} W_{g_0}(y) + \frac{1}{2}(\ln 2 - \frac{1}{2}) = -\frac{1}{2} w_0(y) + \ln 2 - \frac{1}{2}$. Hence for any $k=1,\dots,N$, $G(x,x_k)$ is uniformly bounded on $x\in \C$ such that $\min_{k=1,\dots,N} d_g(x,x_k) > \zeta$. Next, the obvious bound $e^{-2w_0(x)} = (1+|x|^2)^{-2} \leq B (1+|x|)^{-4}$ yields the desired result.

When charge neutrality does not hold, it is easily shown that the regularized correlation functions satisfy the bound, so we focus on the limiting case. We have:
\begin{align*}
    &\ps{ V_\beta(x) \mathbf{V} } = \im e^{-\beta \sum_{k=1}^N \alpha_k G(x,x_k) +2 w_0(x)} (1-e^{2\i \pi s}) e^{P_{\bm z,\bm\alpha}} \int_{0}^{+\infty} \mathbf{E}_{\bm z,\bm\alpha}(-\i\bm c) \dd \bm{c} \\
    &\quad +e^{-\beta \sum_{k=1}^N \alpha_j G(x,x_k) -2 w_0(x)} e^{P_{\bm z,\bm\alpha}} \int_{0}^{2\pi/\beta} \mathbf{E}_{\bm z,\bm\alpha}(\bm c) \dd \bm{c},\quad 
    \mathbf{E}_{\bm z,\bm\alpha}(\bm c) = e^{-\i s\beta  \bm{c}} \ga{ e^{-\mu e^{\i \beta \bm{c}}M_\beta(e^{-\beta^2 G(x,\cdot)} e^{\i \beta H_{\bm{z},\bm{\alpha}}})}} \ .
\end{align*}
We now use Conjecture~\ref{conj:Laplace1}:
\begin{align*}
    \norm{\ps{V_\beta(x)\mathbf{V}}} &\leq e^{-\beta \sum_{k=1}^N \alpha_j G(x,x_k)-2w(x)} e^{P_{\bm{z},\bm \alpha}} \left( 2\norm{\sin{\pi s}} \int_{0}^{+\infty} \sup_{x\in K_\zeta} \norm{ \mathbf{E}_{\bm z,\bm\alpha}(-\i \bm c)} \dd \bm{c} + D\right)
\end{align*}
where $K_\zeta = \{x\in \C, \min_{k=1,\dots,N} d_g(x,x_k) > \zeta\}$. $D>0$ is an upper bound on the integral over the finite horizontal part $[0,2\pi/\beta]$. Because of the global Seiberg bound, $s+\lambda>0$ so
\begin{align*}
\sup_{x\in K_\zeta} \norm{ \mathbf{E}_{\bm z,\bm\alpha}(-\i \bm c)} < e^{-(s+\lambda)\beta \bm{c}} \sup_{\bm{c}\geq 0} \sup_{x\in K_\zeta} e^{(s+\lambda) \beta \bm c} \norm{ \mathbf{E}_{\bm z,\bm\alpha}(-\i \bm c)}
\end{align*}
is an integrable function of $\bm c$, which concludes the proof.
\end{proof}

\subsubsection{Fusion estimates}
Fusion estimates describe the behavior of the correlation functions when two insertion points collide, which in the physics litterature is often referred to as Operator Product Expansion (OPE). They are fundamental analytic inputs in the derivation of the Ward identities.  

\begin{lemma}[Hölder estimate without reflection] \label{lemma:holder} Let $(x_k)_{3\leq k\leq N}\in\C^{N-2}$ be distinct, $\zeta>0$, $K_\zeta\subset\C$ be compact with $\min\limits_{3\leq k\leq N}d(x_k,K_\zeta)>\zeta$. Let $(\beta,\alpha_2,\cdots,\alpha_N)\in\mc A_N$. Then there exists $C\geq 0$ such that for any distinct $x_1, x_2\in K_\zeta$,
\begin{eqs}
    \norm{\ps{V_\beta(x_1) V_{\alpha_2}(x_2) \prod_{j=1}^N V_{\alpha_j}(x_j)}} \leq C |x_1-x_2|^{\beta\alpha_2}.
\end{eqs}
for both regularized and non-regularized correlation functions.
\end{lemma}
\begin{proof}
The proof is similar to the one of Lemma~\ref{lemma:estimate_infinity}. Regarding the prefactor, $x\mapsto G(x_k,x)$ (for $k\geq3$) and $w_0$ are uniformly bounded on the compact space $K_\zeta$. Moreover $W_{g_0}$ is uniformly bounded on $K_\zeta$. Hence for $x_1\neq x_2$ in $K_\zeta$, and some positive $C$,
\begin{align*}
    \norm{ e^{- \beta \alpha_2 G(x_1,x_2)+2w_0(x_1)+2\Delta_{\alpha_2} w_0(x_2)} e^{P_{\bm z,\bm\alpha}} } \leq C |x_1-x_2|^{\beta \alpha_2}.
\end{align*}

It remains to bound the integral over the zero mode. In case of charge neutrality, one can argue that $e^{-\beta^2 G(x_1,x)-\beta \alpha_2 G(x_2,x)}$ is dominated by $e^{-\beta \alpha_2 G(x_2,x)}$, which is itself an integrable function since $\beta \alpha_2 > -2$. Hence as a function of $x_2$, $e^{-\beta \alpha_2 G(x_2,x)}$ is bounded by its $L^1$ norm over $K_\zeta$. The integral formula for moments then gives the required bound.
Without charge neutrality, we rely on Conjecture~\ref{conj:Laplace1}. The function inside the chaos, $x\mapsto e^{-\beta^2 G(x_1,x)-\beta \alpha_2 G(x_2,x)} e^{iH_{\bm x,\bm \alpha}(x)}$, is indeed of the required form, so we can bound the integral by
\begin{align*}
    2\norm{\sin{\frac{\pi}{\beta}\left(\sum_{k=1}^N \alpha_k-2Q\right)}} \int_{0}^{+\infty} \sup_{x_2\in K_\zeta} \norm{ \mathbf{E}_{\bm z,\bm\alpha}(-\i \bm c)} \dd \bm{c} \ .
\end{align*}
\end{proof}
We can further expand this asymptotic based on exponential moments of the GMC measures:
\begin{lemma}[Uniform convergence of exponential moments]
\label{lemma:uniform_convergence_exp_moments}
Let $(f_z)_{z\in U}$, $U\subset\C$ be a family of measurable functions. Further assume that for $A = \D$ and $A=\D^c$
\begin{eqs}
\label{eq:uniform_continuity_exp_moments}
    &\sup\limits_{z\in U} \int_{A^2} |f_z(x)| |f_z(y)| |x-y|^{-\beta^2} \dd v_g(x) \dd v_g(y) < \infty\qt{and}\\
    &\sup\limits_{z\in U} \int_A |f_z(x)| \mathrm{dist}(x,\partial \D)^{-\frac{\beta^2}{2}} \dd v_g(x) < \infty.\qt{Then}\sup\limits_{z\in U} \ga{e^{|M_\beta^g(f_z)|}} < \infty \ .
\end{eqs}
\end{lemma}
\begin{proof}
    This is a mere consequence of the proof of Lemma~\ref{lemma:def_laplace}.
\end{proof}

\begin{lemma}[Second order OPE]
\label{lemma:fusion_general}
Let $(\alpha_k)_{1\leq k\leq N}\in\mc A_N$ with $m-2< \beta(\alpha_1+\alpha_2)<m-1$ for some $m\in\N$. Further assume that $(\alpha_1+\alpha_2,\alpha_3,\cdots,\alpha_N,\underbrace{\beta,\cdots,\beta}_{k\text{ terms}})\in\mc A_{N-1+k}$ for $0\leq k\leq m$ and that $(\alpha_1+\alpha_2+\beta,\alpha_3,\cdots,\alpha_N )\in\mc A_{N-1}$. Then, for any $(x_k)_{3\leq k\leq N}\in\C^{N-2}$ distinct, there exist constants $C_{i,j}$, $1\leq i+j\leq m$, such that for $x\in\C\setminus\{0,x_3,\cdots,x_N\}$,
\begin{eqs}\label{eq:OPE_2}
    \norm{x}^{-\alpha_1 \alpha_2}&\ps{ V_{\alpha_1}(0) V_{\alpha_2}(x) \mathbf{V} } = \ps{ V_{\alpha_1+\alpha_2}(0) \mathbf{V} } +\sum_{1\leq i+j\leq m}C_{i,j}x^i\bar x^j\\
    &+ |x|^{\beta(\alpha_1+\alpha_2)+2} \rho(\beta\alpha_1,\beta\alpha_2) \ps{V_{\alpha_1+\alpha_2+\beta}(0) \mathbf{V}} + \underset{x\to0}{o}(|x|^{\beta(\alpha_1+\alpha_2)+2})
\end{eqs}
where $\mathbf{V} = \prod_{k=3}^N V_{\alpha_k}(x_k)$ and, with $\gamma(x) = \Gamma(x)/\Gamma(1-x)$,
\begin{eqs}
    &\rho(a,b) =\frac{\mu\pi}{\gamma\left(-\tfrac{a}{2}\right) \gamma\left(-\tfrac{b}{2}\right) \gamma\left(2-\tfrac{a+b}{2}\right) }\cdot
\end{eqs}
\end{lemma}
\begin{proof}
Let $f$ be continuous near $0$, and such that the family of functions
\[
    f_x = e^{-\beta \alpha_1 G(\cdot,0) - \beta \alpha_2 G(\cdot,x)} f(\cdot)
\]
satifies the condition~\eqref{eq:uniform_continuity_exp_moments} for $x\in B(0,\eta)$ for some $\eta>0$.  For such a $x$, we set
\[
    I(x) = \int_{\C} f_x(y) M_\beta^g(\dd^2 y) = e^{\frac{\beta\alpha_2}2\left(w_0(x)-w_0(0)\right)}\int_{\C} \norm{1-\frac xy}^{\beta\alpha_2}f_0(y) M_\beta(\dd^2 y).
\]
For any $n\geq0$, let us set $P_{n}(\xi)\coloneqq \sum_{0\leq i+j\leq n}p_{i,j}\xi^i\bar\xi^j$ so that $\norm{1-\xi}^{\beta\alpha_2}=P_n(\xi)+o(\norm{\xi}^{n})$ as $\xi\to0$. Since $\beta(\alpha_1+\alpha_2)>m-2>\beta Q+m-1$, $\int_{\C} P_{m-1}\left(\frac xy\right)f_0(y) M_\beta(\dd^2 y)$ is well-defined and has exponential moments by Lemma~\ref{lemma:uniform_convergence_exp_moments}. Hence we can write for any $\mu\in\C$
\begin{align*}
    &\expect{e^{-\mu I(x)}}=\expect{e^{-\mu I_1(x)}}-\mu \expect{I_2(x)e^{-\mu I_1(x)}}+\expect{e^{-\mu I_1(x)}\left(e^{-\mu I_2(x)}-1+\mu I_2(x)\right)},\qt{with}\\
    &I_1(x)\coloneqq e^{\frac{\beta\alpha_2}2\left(w_0(x)-w_0(0)\right)}\int_{\C} P_{m-1}\left(\frac xy\right)f_0(y) M_\beta(\dd^2 y)\qt{and} I_2(x)\coloneqq I(x)-I_1(x).
\end{align*}
Now since $I_1(x)$ has exponential moments, $x\mapsto \expect{e^{-\mu I_1(x)}}$ is smooth at $x=0$, with $\partial_x^p\partial_{\bar x}^q\expect{e^{-\mu I_1(x)}}$ a linear combination of $\expect{\left(\mu I_1(x)\right)^re^{-\mu I_1(x)}}$ for $r\leq p+q$.
As a consequence, as soon as $(\alpha_1+\alpha_2,\alpha_3,\cdots,\alpha_N,\beta,\cdots,\beta)\in\mc A_{N-1+m}$:
\[
    \int_\mc{U} e^{\im s\bm c}\expect{\expo{-\mu e^{\im \beta\bm c}I_1(x)}}(\gamma_0)_\ast(\dd \bm c) =\sum_{0\leq i+j\leq m}C_{i,j}x^i\bar x^j+\mc O\left(\norm{x}^{m+1}\right).
\]
Let us now turn to the second term in the expansion. We can first write
\begin{align*}
    &\expect{I_2(x)e^{-\mu I_1(x)}}=e^{\frac{\beta\alpha_2}2\left(w_0(x)-w_0(0)\right)}\int_{\C} \left(\norm{1-\frac xy}^{\beta\alpha_2}-P_{m-1}\left(\frac xy\right)\right)f_0(y) \expect{V_\beta(y) e^{-\mu I_1(x)}}\norm{\dd y}^2\\
    &=e^{\frac{\beta\alpha_2}2\left(w_0(x)-w_0(0)\right)}\int_{\C} \sum_{i+j=m}c_{i,j}\left(\frac xy\right)^i\left(\frac{\bar x}{\bar y}\right)^jf_0(y) \expect{V_\beta(y) e^{-\mu I_1(x)}}\norm{\dd y}^2\\
    &+\norm{x}^{2+\beta(\alpha_1+\alpha_2)}e^{\frac{\beta\alpha_2}2\left(w_0(x)-w_0(0)\right)}\int_{\C} \left(\norm{1-\frac 1y}^{\beta\alpha_2}-P_{m}\left(\frac 1y\right)\right)\frac{f_0(xy)}{\norm{x}^{\beta(\alpha_1+\alpha_2)}} \expect{V_\beta(xy) e^{-\mu I_1(x)}}\norm{\dd y}^2.
\end{align*}
The first integral is absolutely convergent since $i+j=m<2+\beta(\alpha_1+\alpha_2)$. Changing $\mu\to\mu e^{-\beta\bm c}$ and integrating $\bm c$ between $0$ and $+\infty$ (the integral converges as soon as $(\alpha_1+\alpha_2,\alpha_3,\cdots,\alpha_N,\beta)\in\mc A_N$) yields
\begin{align*}
    &\int_0^{+\infty}e^{-\beta s \bm c}\expect{I_2(x)e^{-\mu e^{-\beta\bm c}I_1(x)}}\dd\bm c\\
    &=\sum_{i+j=m}C'_{i,j}x^i\bar x^j+\mc O\left(\norm{x}^{m+1}\right)+\norm{x}^{2+\beta(\alpha_1+\alpha_2)}(1+o(1))\rho_m(\beta\alpha_1,\beta\alpha_2)\tilde f_0(0) \expect{V_\beta(0) e^{-\mu I_1(0)}}
\end{align*}
where $\tilde f_0(0)=e^{\beta (\alpha_1+\alpha_2) W_g(0)} f(0)$, and $\rho_m(a,b) = \int_\C |v|^{a+b} \left(\norm{1-\frac1v}^{b}-P_m\left(\frac1v\right)\right) \norm{\dd v}^2$. Of course the same reasoning is valid when $\bm c$ is integrated in $(0,\frac{2\pi}\beta)$.
Recollecting terms, we see that if we choose $f(y)=e^{-\beta\sum_{k=3}^N\alpha_kG(\cdot,z_k)}$ we obtain the expansion in Equation~\eqref{eq:OPE_2} provided that we identify the terms $\norm{z}^{\eta}$ in the expansion for $\eta=0$ and $\eta=2+\beta(\alpha_1+\alpha_2)$.

For the constant order term in the expansion, it is given by $\lim\limits_{x\to0}\norm{x}^{-\alpha_1\alpha_2}\ps{V_{\alpha_1}(0)V_{\alpha_2}(x)\V}$. Writing $e^{-\alpha_1\alpha_2G(x,0)}=\norm{x}^{\alpha_1\alpha_2}e^{\frac{\alpha_1\alpha_2}{2}\left(W_g(x)+W_g(0)\right)}$ (see Equation~\eqref{eq:green}), and recalling the definition of $P_{\bm z,\bm\alpha}$ from Equation~\eqref{eq:def_HP}, we obtain that $C_{0,0}=\ps{V_{\alpha_1+\alpha_2}(0)\V}$ as expected.
Let us now turn to the term corresponding to $\norm{x}^{2+\beta(\alpha_1+\alpha_2)}$. By the same reasoning the term $\tilde f_0(0) \expect{V_\beta(0) e^{-\mu I_1(0)}}$ will give rise to $\ps{V_{\alpha_1+\alpha_2+\beta}(0)\V}$, while the prefactor $\rho_m(a,b)$ can be evaluated as follows: for $m-2<a+b<m-1$
\begin{eqs}
    \rho_m(a,b) =\frac{\pi}{\gamma\left(-\tfrac{a}{2}\right) \gamma\left(-\tfrac{b}{2}\right) \gamma\left(2-\tfrac{a+b}{2}\right) } \ \cdot
\end{eqs}
To see why note the the following function is meromorphic in $-2-\Re(b)<\Re(\cdot)<m-1-\Re(b)$
\[
    F(a)\coloneqq \int_\C |v|^{a+b} \left[\norm{1-\frac1v}^{b}-P_m\left(\frac1v\right)\mathds 1_{\norm{v}>1}\right] \norm{\dd v}^2 -2\pi\sum_{0\leq i\leq \lfloor m/2\rfloor}p_{i,i}\frac1{2+a+b-2i}
\]
with poles at $-b-2,-b-4,\cdots,-b-2m$. Moreover, for $-2-\Re(b)<\Re(a)<-1-\Re(b)$, by explicit computations (the second equality follows from~\cite[p. 504]{forresterImportanceSelbergIntegral2008})
\[
    F(a)= \int_\C |v|^{a}\norm{v-1}^{b} \norm{\dd v}^2=\frac{\pi}{\gamma\left(-\tfrac{a}{2}\right) \gamma\left(-\tfrac{b}{2}\right) \gamma\left(2-\tfrac{a+b}{2}\right) }
\]
while for $m-2-\Re(b)<\Re(\cdot)<m-1-\Re(b)$ it coincides with $\rho_m(a,b)$.
\end{proof}

\subsubsection{Boundary terms} \label{section:boundary_terms} As mentioned in the introduction, it is crucial to control the potential divergences arising when the regularization of the domain $\C_\rho$ is lifted. In particular, in the derivation of Ward identities, integrals on the boundary components of the domain $\C_\rho$ occur. The fusion estimate derived in Lemma~\ref{lemma:fusion_general} allows to control such boundary terms.

For instance, defining $\L_{-1}V_{\alpha_1}(x_1)$ within correlation functions involves the $\rho\to0$ limit of
\begin{eqs}
\label{eq:boundary_L1}
    &\rho^{\beta\alpha_1}\oint_{\partial B(x_1,\rho)} \mathbf{G}(z) \i\dd \overline{z},\quad\mathbf{G}(z) \coloneqq \norm{z-x_1}^{-\beta\alpha_1}\ps{ V_\beta(z) \prod_{k=1}^N V_{\alpha_k}(z_k)}.
\end{eqs}

If $\beta\alpha_1>-1$ this integral vanishes as $\rho\to0$ using Lemma~\ref{lemma:holder}. If $\beta\alpha_1\leq -1$, from Lemma~\ref{lemma:fusion_general}
\[
    \rho^{\beta\alpha_1}\oint_{\partial B(x_1,\rho)} \mathbf{G}(z) \i\dd \overline{z}=\rho^{\beta\alpha_1}\oint_{\partial B(x_1,\rho)} \left(\mathbf{G}(z)-\mathbf{G}(x_1)\right)\i\dd \overline{z}=\mc O\left(\rho^{1+\beta\alpha_1+\min(1,2+\beta(\beta+\alpha_1))}\right),
\]
which vanishes as soon as $(\alpha_1+j\beta,\alpha_2,\cdots,\alpha_N)\in\mc A_{N}$ for $j=0,1,2$. Hence, under this additional assumption, the contour integral vanishes in the limit $\rho\to0$ for $\alpha_1>-\frac3{2\beta}-\frac\beta2$.

To define $\L_{-n}V_{\alpha_1}(x_1)$ with $n\geq 2$, the boundary integral is
\begin{eqs}
\label{eq:boundary_Ln}
    &\rho^{\beta\alpha_1}\oint_{\partial B(x_1,\rho)} \frac{1}{(z-x_1)^{n-1}} \mathbf{G}(z)\i \dd \overline{z}.
\end{eqs}
Using the same reasoning as above, this integral is a $\mc O\left(\rho^{2-n+\alpha_1\beta+\min(n,2+\beta(\beta+\alpha_1))}\right)$ under the additional assumptions that $(\alpha_1+\beta,\cdots,\alpha_{N},\underbrace{\beta,\cdots,\beta}_{k\text{ terms}})\in\mc A_{N+k}$ for $0\leq k\leq n-1$ and $(\alpha_1+2\beta,\cdots,\alpha_{N})\in\mc A_{N}$. Provided both these requirements are met,
the integral thus vanishes in the $\rho\to0$ limit as soon as $\alpha_1 > \frac{n-4}{2\beta} -\frac{\beta}2=Q+\frac{n}{2\beta}-\beta$.

\subsubsection{KPZ identity}
As will become apparent later, KPZ identity is crucial in the derivation of the Ward identities. It is because of the KPZ identity that metric-dependent terms occuring in the derivation of Ward identities vanish.

\begin{proposition}[KPZ identity] \label{KPZ} Let $\bm\alpha\in\mc A_N$ be such that $(\alpha_1,\cdots,\alpha_N,\beta)\in\mc A_{N+1}$, and set $s = \frac{2Q-\sum_{k=1}^N \alpha_k}{\beta}$. Then, under the assumptions of Conjecture~\ref{conj:Laplace1}, for any $\rho>0$,
\begin{equation}
\int_{\C_\rho} \ps{ V_{\beta}(x) \prod_{k=1}^N V_{\alpha_k}(x_k) }_\rho \norm{\dd x}^2 = -\frac{s}{\mu} \ps{ \prod_{k=1}^N V_{\alpha_k}(x_k) }_\rho \ .
\end{equation}
\end{proposition}
\begin{proof}
We start from the definition of the regularized correlation functions~\eqref{eq:correl_der}:
\begin{align*}
	\ps{\prod_{k=1}^N V_{\alpha_k}(x_k)}_\rho &= \im \left(1-e^{-2\i\pi s }\right)e^{P_{\bm z,\bm\alpha}}\int_{0}^{+\infty} e^{-s\beta \bm{c}} \expect{\expo{-\mu e^{\beta\bm c}M_{\beta}\left(\C_\rho,e^{\i\beta H_{\bm x,\bm\alpha}}\right)}}\dd \bm{c}\\
	&\quad + e^{P_{\bm z,\bm\alpha}} \int_{0}^{\frac{2\pi}{\beta}} e^{-\im s\beta \bm{c}} \ga{\expo{-\mu e^{\im \beta \bm{c} }M_{\beta}\left(\C_\rho,e^{\i\beta H_{\bm x,\bm\alpha}}\right)}} \dd \bm{c}
\end{align*}
Because the Laplace transform is analytic, and by Conjecture~\ref{conj:Laplace1} and Lemmas~\ref{lemma:estimate_infinity} and~\ref{lemma:holder} the integrals over $\bm c$ and $x$ converge absolutely, we can take the derivative with respect to $\mu$ in the right hand side and exchange the order of integration:
\begin{equation*}
	\partial_\mu \ps{\prod_{k=1}^N V_{\alpha_k}(x_k)}_\rho = - \int_{\C}\ps{ V_{\beta}(x) \prod_{k=1}^N V_{\alpha_k}(x_k) }_\rho \; \norm{\dd x}^2 \ .
\end{equation*}
On the other hand, the dependency in $\mu$ can be made explicit using the $\mc U$ contour 
\begin{align*}
    \ps{\prod_{k=1}^N V_{\alpha_k}(x_k) }_\rho = e^{P_{\bm z,\bm\alpha}} \int_{\mathcal{U}} e^{\i s\beta \bm c} \expect{\expo{-\mu e^{\i \beta\bm c} M_{\beta}\left(\C_\rho,e^{\i\beta H_{\bm x,\bm\alpha}}\right)}}(\gamma_0)_*\left(\dd \bm{c}\right)
\end{align*}
and changing variables $\bm{c} \to \bm{c} - \frac{1}{\beta} \ln\mu$. The integration contour of the zero mode is now shifted in the imaginary direction. This amounts to add the contribution of a closed contour to the original integral, and since the integrand is holomorphic in $\bm{c}$, the value of the integral is not modified. Hence we can conclude as expected:
\begin{equation*}
	\ps{ \prod_{k=1}^N V_{\alpha_k}(x_k) }_\rho = \mu^s \left[\ps{ \prod_{k=1}^N V_{\alpha_k}(x_k) }_\rho\right]_{\mu=1} \Rightarrow
	\partial_\mu \ps{ \prod_{k=1}^N V_{\alpha_k}(x_k) }_\rho = \frac{s}{\mu} \ps{ \prod_{k=1}^N V_{\alpha_k}(x_k) }_\rho.
\end{equation*}
\end{proof}

\subsection{Proof of Ward identities}
Ward identities involve derivatives of the Liouville field, and as such need to undergo a regularization procedure. This procedure allows to define derivatives of the correlation functions, and more generally of the insertion of descendant fields within them. To be more specific, we wish to define for $n\geq 1$
\begin{align*}
    \ps{\L_{-n}V_\alpha(x)\prod_{k=1}^NV_{\alpha_k}(x_k)}\coloneqq\lim\limits_{\rho,\eps\to0}\ps{\L_{-n}^\alpha(x)V_\alpha(x)\prod_{k=1}^NV_{\alpha_k}(x_k)}_{\rho,\eps}
\end{align*}
where recall $\L_{-n}^\alpha$ from Equation~\eqref{eq:def_desc}, and with the right-hand side defined in Definition~\ref{def:der_cor} with $R=+\infty$. 

\subsubsection{Warmup: first order descendants and the derivability of correlation functions}
We prove the following proposition, that will serve as an exposition of our methods:
\begin{proposition}[Derivability of correlation functions]
\label{prop:derivability}
Assume that $\alpha_1>-\frac{3}{2\beta}-\frac{\beta}{2}$ and that $(\alpha_1,\cdots,\alpha_N)\in\mc A_{N}$ and $(\alpha_1,\cdots,\alpha_N,\beta)\in\mc A_{N+1}$. Then the following limit exists and we have, in the sense of weak derivatives,
\begin{equation}
    \ps{\L_{-1} V_{\alpha_1}(x_1) \prod_{k=2}^N V_{\alpha_k}(x_k)} \coloneqq \lim\limits_{\rho,\eps\to0}\ps{\L_{-1}^{\alpha_1}(x_1)\prod_{k=1}^NV_{\alpha_k}(x_k)}_{\rho,\eps} = \partial_{x_1} \ps{\prod_{k=1}^N V_{\alpha_k}(x_k)}.
\end{equation}
\end{proposition}
Before moving to the proof we give the following statement, key to define the descendants:
\begin{lemma}[Insertion of a derivative of the Liouville field]
\label{lemma:insertions}
Let $(x_k)_{1\leq k\leq N}\in\C^N$ be distinct, $\bm\alpha\in\mc A_N$, and let  $\mathbf{V} = \prod_{k=1}^N V_{\alpha_k}(x_k)$. For any $p\geq 1$, let $F_p:f\mapsto \frac{\partial^pf}{(p-1)!}$. Then
\begin{equation}
    \ps{F_p(x_1)\V}_\rho= \sum_{k=2}^N \frac{ \i \alpha_k}{2(x_k-x_1)^p} \ps{\V}_\rho-\mu  \int_{\C_\rho} \frac{\i\beta}{2(z-x)^p}  \ps{ V_{\beta}(z) \mathbf{V} }_\rho \, \norm{\dd z}^2.
\end{equation}
\end{lemma}
\begin{proof}
Recall that $\lioureg = \X^g_{\eps} + \i Q w_0 + \bm{c}$. Since the integral on $\bm{c}$ converges absolutely, and by definition of the derivatives of the Liouville field~\ref{def:der_cor},
\begin{align*}
&\ps{F_p(x_1) \mathbf{V} }_{\rho,\eps} = \int_{\mc U} e^{ \i\bm s \bm{c} } \; \E\left[ \frac{\partial^p}{(p-1)!}\left( \X^g_\eps + \i Q w_0 + H^{(1)}_{\bm z,\bm\alpha}\right)(x_1)e^{-\mu e^{\i\beta \bm{c} } M_{\beta,\eps}(\C_\rho,e^{\i\beta H_{\bm x,\bm\alpha}}) } \right]\, (\gamma_0)_*\left(\dd \bm{c}\right) \\
&=\frac{\partial^p}{(p-1)!}\left(H^{(1)}_{\bm z,\bm\alpha} + \i Q w_0 \right)(x_1)\ps{\V}_\rho+\int_{\mc U} e^{ \i\bm s \bm{c} } \; \E\left[ \frac{\partial^p}{(p-1)!} \X^g_\eps(x_1)e^{-\mu e^{\i\beta \bm{c} } M_{\beta,\eps}(\C_\rho,e^{\i\beta H_{\bm x,\bm\alpha}}) } \right]\, (\gamma_0)_*\left(\dd \bm{c}\right)\ .
\end{align*}
This last term can be dealt with using Gaussian integration by parts (Lemma~\ref{gaussian_integration}) by writing the regularized GMC mass as a Riemann sum. Using the expression of the Green's function~\eqref{eq:green} we get
\begin{align*}
	\ps{F_p(x_1) \mathbf{V} }_{\rho,\eps} &= \sum_{k=1}^N \frac{\i \alpha_k}{2(x_k-x_1)^p}\ps{\mathbf{V} }_{\rho,\eps}- \mu\int_{\C_\rho}\frac{\i \beta}{2(z-x_1)^p}\ps{V_\beta(z)\mathbf{V} }_{\rho,\eps}\norm{\dd x}^2\\
    &-\frac{\i\partial^pw_0(x)}{2(p-1)!}\left(\sum_{k=1}^N\alpha_k-2Q-\mu\beta\int_{\C_\rho}\ps{V_\beta(z)\V}_{\rho,\eps}\norm{\dd x}^2\right) \ .
\end{align*}
Taking $\eps\to0$, the second line vanishes thanks to Lemma~\ref{KPZ}, concluding the proof.
\end{proof}

We are now in position to prove Proposition~\ref{prop:derivability}.
\begin{proof}
To start with, note that by definition of $\L_{-1}^\alpha$ we have for $\rho,\eps>0$
\begin{equation}\label{eq:derivatives}
     \ps{\L_{-1}^{\alpha_1}(x_1)\V}_{\rho,\eps}=\partial_{x_1}\ps{\V}_{\rho,\eps}+o_\eps(1).
\end{equation}
To define the left-hand side, we take $p=1$ in Lemma~\ref{lemma:insertions}, which gives
\[
    \ps{\L_{-1}^{\alpha_1}(x_1)\V}_{\rho}=\sum_{k=1}^N \frac{\alpha_1 \alpha_k}{2(x_1-x_k)} \ps{\mathbf{V} }_{\rho} -\mu \int_{\C_\rho} \frac{ \alpha_1 \beta}{2(x_1-z)} \ps{ V_{\beta}(z) \mathbf{V} }_{\rho} \norm{\dd z}^2+o_\eps(1).
\]
Clearly, the sum on the RHS stays finite in this limit. Regarding the integral, the asymptotic behavior of the correlation functions~\eqref{eq:asymp_estimates} prove that singularities may arise only when $z$ approaches insertion points. However, by Hölder estimates (see Lemma~\ref{lemma:holder}), the divergence as $z$ approaches $x_k$ for $2\leq k\leq N$ is dominated by $|z-x_k|^{\beta \alpha_k}$, which is integrable since $\beta\alpha_k > \beta Q>-2$. The only remaining case is when $z$ approaches $x_1$. We split the integration domain into two parts: if $0<\eta<\frac{1}{2}\min_{k=2,\dots,N} |x_1-x_k|$, let $A_{\eta,\rho}=B\left(x,\eta \right)\cap \C_\rho$, so that the integral over its complementary in $\C_\rho$ is absolutely convergent. Then by Eq.~\eqref{eq:derivatives}
\begin{align*}
    \partial_{z} \ps{V_{\beta}(z)\mathbf{V}}_{\rho} = \sum_{k=1}^N \frac{\beta \alpha_k}{2(x_k-z)} \ps{V_{\beta}(z)\mathbf{V}}_{\rho} -\mu \int_{\C_\rho} \frac{\beta^2}{2(z^\prime-z)} \ps{V_{\beta}(z)V_{\beta}(z^\prime)\mathbf{V}}_{\rho} \norm{\dd z^\prime}^2.
\end{align*}
As a consequence
\begin{align*}
    &\int_{A_{\eta,\rho}} \frac{\alpha_1\beta}{2(x_1-z)} \ps{ V_{\beta}(z)\mathbf{V} }_\rho \norm{ \dd z}^2= -\sum_{k=2}^N \int_{A_{\eta,\rho}} \frac{\beta\alpha_k}{2(x_k-z)} \ps{V_{\beta}(z) \mathbf{V} }_{\rho}\norm{ \dd z}^2 \\
    &+ \oint_{\partial A_{\eta,\rho}} \ps{ V_{\beta}(z) \mathbf{V} }_{\rho} \frac{\i\dd \overline{z}}2 + \mu \int_{A_{\eta,\rho}} \int_{\C_\rho} \frac{\beta^2}{2(z-z^\prime)} \ps{ V_{\beta}(z) V_{\beta}(z^\prime) \mathbf{V} }_{\rho} \norm{ \dd z}^2\norm{ \dd z^\prime}^2 \ .
\end{align*}
The first integral remains finite as the regularization is lifted since we stay at a finite distance from the $x_k$. Before dealing with the integral over $\partial A_{\eta,\rho}$, we consider the two-fold integral:
\begin{align*}
    &\int_{A_{\eta,\rho}} \int_{\C_\rho} \frac{1}{(z-z^\prime)} \ps{V_{\beta}(z) V_{\beta}(z^\prime) \mathbf{V} }_\rho \norm{\dd z}^2 \norm{\dd z^\prime}^2= \\
    &\int_{A_{\eta,\rho}} \int_{\C_\rho\setminus A_{\eta,\rho}} \frac{1}{(z-z^\prime)}  \ps{V_{\beta}(z) V_{\beta}(z^\prime) \mathbf{V} }_\rho \norm{\dd z}^2 \norm{\dd z^\prime}^2+ \int_{(A_{\eta,\rho})^2} \frac{1}{(z-z^\prime)}  \ps{V_{\beta}(z) V_{\beta}(z^\prime) \mathbf{V} }_\rho \norm{\dd z}^2 \norm{\dd z^\prime}^2
\end{align*}
Thanks to the fusion estimates~\ref{lemma:fusion_general}, singularities occur only when $z\to z^\prime$, which happens near $\partial B(x,\eta)$ hence at fixed distance from the insertions $x_k$. For such fusions, the integrand is dominated by $|z-z^\prime|^{\beta^2-1}$, which is integrable so the limit $\rho\to 0$ is finite. For the second integral, due to the antisymmetry of the integrand it is identically zero. 
For the integral over $\partial A_{\eta,\rho}$, it is the disjoint union of $\partial B(x,\eta)$ and $\partial B(x,\rho)$. The integral over $\partial B(x,\eta)$ remains uniformly bounded in $\rho$, while the one over $\partial B(x,\rho)$ vanishes in the limit $\rho \to 0$ by the discussion in~\ref{section:boundary_terms}, concluding the proof.
\end{proof}

\subsubsection{Local Ward identities for descendants of arbitrary order}
We now turn to prove Ward identities for descendants $\L_{-n} V_\alpha$ (Thm.~\ref{thm:desc}). Correlation functions will contain Wick products of the GFF, so we begin by stating the following lemma:
\begin{lemma}[Insertion of a Wick product of derivatives]
\label{lemma:wick_insertions}
With the same notations as Lemma~\ref{lemma:insertions}, for any sequence of positive integers $r,p_1,\dots,p_r$,
\begin{eqs}
\label{eq:wick_ibp}
	\frac1{(p_r-1)!}&\ps{ :\prod_{l=1}^r \partial^{p_l} \liou(x_1): \mathbf{V} }_{\rho} = \sum_{k=2}^N \frac{\i \alpha_k}{2(x_k-x_1)^{p_r}} \ps{ :\prod_{l=1}^{r-1} \partial^{p_l} \Phi(x_1): \mathbf{V} }_{\rho}  \\
    &- \mu  \int_{\C_\rho} \frac{\i\beta}{2(z-x_1)^{p_r}}  \ps{:\prod_{l=1}^{r-1} \partial^{p_l} \Phi(x): V_\beta(z) \mathbf{V} }_\rho \norm{\dd z}^2.
\end{eqs}
\end{lemma}
\begin{proof}
The proof is the same as Lemma~\ref{lemma:insertions} using the recursive definition of the derivatives of the Liouville field~\ref{eq:def_deri} together with the Gaussian integration by parts formula (Eq.~\ref{gaussian_integration}).
\end{proof}

\begin{proof}[Proof of Theorem~\ref{thm:desc}]
Applying Lemmas~\ref{lemma:insertions} and~\ref{lemma:wick_insertions} for regularized correlation functions,
\begin{align*}
    &(-1)^n\psrreg{\L^{\alpha_1}_{-n}(x_1) \mathbf{V}} = \left(\sum_{k=2}^N \frac{-((n-1)Q+\alpha_1)\alpha_k}{2(x_1-x_k)^n} + \sum_{j=1}^{n-1} \sum_{k,l=2}^N \frac{\alpha_k\alpha_l}{4(x_1-x_k)^j (x_1-x_l)^{n-j}}\right) \psrreg{ \mathbf{V}} \\
    &\quad-\mu\int_{\C_\rho}\left( \sum_{j=1}^{n-1} \sum_{k=2}^N \frac{ \beta \alpha_k}{2(x_1-x_k)^j(x_1-z)^{n-j}}+\frac{n-1-\frac{\beta \alpha_1}2}{(x_1-z)^n}\right)\psrreg{V_\beta(z) \mathbf{V}} |\dd^2 z|\\
    &\quad + \mu^2\sum_{k=1}^{n-1}  \int_{\C_\rho} \int_{\C_\rho} \frac{\beta^2}{4(x_1-z)^{k}(x_1-z^\prime)^{n-k}} \psrreg{V_\beta(z)V_\beta(z^\prime) \mathbf{V}} |\dd^2 z| |\dd^2 z^\prime|.
\end{align*}
The two fold integral is rewritten by means of the following identity:
\begin{align*}
    \partial_z\left(\frac{1}{(x_1-z) ^{n-1}} \psrreg{V_\beta(z)\mathbf{V}}\right)=& \left(\frac{n-1-\frac{\beta\alpha_1}2}{(x_1-z) ^n} + \sum_{k=2}^N \frac{\beta \alpha_k}{2(x_1-z)^{n-1} (z-x_k) }\right)\psrreg{V_\beta(z) \mathbf{V}} \\
    &\quad-\mu\int_{\C_\rho}  \frac{\beta^2}{2(x_1-z)^{n-1}(z^\prime-z) } \psrreg{V_\beta(z)V_\beta(z^\prime)\mathbf{V}} |\dd^2 z^\prime|
\end{align*}
along with the symmetrization identities
\begin{align*}
    &\sum_{j=1}^{n-1} \sum_{k,l=2}^N \frac{1}{(x_1-x_k)^j(x_1-x_l)^{n-j}} = \sum_{k=2}^N \frac{n-1}{(x_1-x_k)^n} - \sum_{\substack{k,l=2\\
    k\neq l}}^N \frac{2}{(x_k-x_l)(x_1-x_k)^{n-1}}\\
    &\sum_{j=1}^{n-1} \sum_{k=2}^N \frac{1}{(x_1-x_k)^j(x_1-z)^{n-j}} = \sum_{k=2}^N \frac{1}{z-x_k}\left(\frac{1}{(x_1-z)^{n-1}}-\frac{1}{(x_1-x_k)^{n-1}}\right) \\
    &\sum_{j=1}^{n-1} \frac{1}{(x_1-z)^j}\frac{1}{(x_1-z^\prime)^{n-j}} = \frac{1}{z^\prime-z} \left(\frac{1}{(x_1-z)^{n-1}}-\frac{1}{(x_1-z^\prime)^{n-1}}\right) \ .
\end{align*}
Hence
\begin{align*}
    &(-1)^n\psrreg{\L^{\alpha_1}_{-n}(x_1) \mathbf{V}}=\sum_{k=2}^N \frac{\Delta_{\alpha_k}}{(x_1-x_k)^n }\ps{\V}_\rho \\
    & -\sum_{k=2}^N \frac{1}{(x_1-x_k) ^{n-1}}\left(  \sum_{\substack{l=1\\l\neq k}}^N \frac{\alpha_k\alpha_l}{2 (x_l-x_k)  } \psrreg{\mathbf{V}}-\mu \int_{\C_\rho} \frac{\alpha_k\beta}{2(z-x_k) }\psrreg{V_\beta(z) \mathbf{V}} |\dd z|^2\right)\\
    &-\int_{\C_\rho} \partial_z\left(\frac{1}{(x_1-z) ^{n-1}} \psrreg{V_\beta(z)\mathbf{V}}\right) |\dd z|^2.
\end{align*}
On the second line, we recognize correlation functions with the descendant $\L_{-1}V_{\alpha_k}(x_k)$:
\[
     \ps{\L_{-1}^{\alpha_k}(x_k)\V}_\rho=\sum_{\substack{l=1\\l\neq k}}^N \frac{\alpha_k\alpha_l}{2 (x_l-x_k)  } \psrreg{\mathbf{V}}-\mu \int_{\C_\rho} \frac{\alpha_k\beta}{2(z-x_k) }\psrreg{V_\beta(z) \mathbf{V}} |\dd z|^2.
\]
As a consequence to conclude for the proof it remains to show that the last term goes to zero as we let $\rho\to0$. It is equal to a sum of one-dimensional integrals along the boundary components of $\C_\rho$,
\begin{align*}
\partial\C_\rho = \partial B(0,\rho^{-1}) \cup \bigcup_{k=1}^N \partial B(x_k,\rho) \ .
\end{align*}
Given $\rho <\frac{1}{2}\min \{|x_k-x_1|,\,k=2,\dots,N\}\cup\{x_1^{-1}\}$, we stay away from the singularities of the integrand of
\begin{equation*}
    \int_{\partial B(x_k,\rho)} \frac{1}{(x_1-z)^{n-1} } \psrreg{V_\beta(z) \mathbf{V}} \im \dd \overline{z} 
\end{equation*}
when $k\geq 2$. By doing an expansion at order $1$ as $z\to x_k$ in the same fashion as in the definition of the $\L^{\alpha}_{-1}V_{\alpha}$ descendant we see that this term vanishes in the limit as soon as $\ps{\L_{-1}^{\alpha_k}(x_k)\V}$ is well-defined, hence the condition on the $\alpha_k$. The integral over $B(0,\rho^{-1})$ is controlled by the large $z$ asymptotics and vanishes similarly. The last remaining term is the integral over the circle $\partial B(x_1,\rho)$. By the discussion of Section~\ref{section:boundary_terms}, it vanishes in the limit $\rho\to 0$ in virtue of our hypothesis $\alpha > \frac{n-4}{2\beta} - \frac\beta2$, concluding the proof.
\end{proof}

\subsubsection{Second order descendants and BPZ equation}
Second order descendants involve vertex operators $\L_{-2}V_\alpha$ but also compositions of descendants of order 1: $\L_{(-1,-1)}V_\alpha[\Phi]$. These are defined by considering for $f\in C^\infty(\C,\C)$
\begin{equation}
    \L_{(-1,-1)}^\alpha:f\mapsto \i\alpha \partial^2f - \alpha^2 (\partial f)^2.
\end{equation}

\begin{proposition}[Second derivative of correlation functions]
\label{prop:second_derivative}
Assume that $\alpha_1>-\frac{\beta}{2}-\frac{1}{\beta}$ and that $(\alpha_1,\cdots,\alpha_N)\in\mc A_{N}$ and $(\alpha_1,\cdots,\alpha_N,\beta)\in\mc A_{N+1}$. Then the following limit exists and we have, in the sense of weak derivatives,
\begin{equation}
    \ps{\L_{(-1,-1)}^{\alpha_1}(x_1) \prod_{k=1}^N V_{\alpha_k}(x_k)} = \partial^2_x \ps{ V_\alpha(x) \prod_{j=1}^N V_{\alpha_j}(x_j)}\,.
\end{equation}
\end{proposition}
\begin{proof}
By Lemma~\ref{lemma:wick_insertions} and the definition of Wick products, it is straightforward to see that
\begin{align*}
    \partial^2_{x_1} \psrreg{ \mathbf{V} } &= \im \alpha_1 \psrreg{ \partial^2 \Phi(x_1) \mathbf{V} } - \alpha_1^2 \psrreg{ (\partial \Phi(x_1))^2 \mathbf{V} } \ ,
\end{align*}
proving the claim at the level of regularized correlation functions. Next, as in the proof of Thm~\ref{thm:desc}, we use Lemmas~\ref{lemma:insertions} and~\ref{lemma:wick_insertions} to write explicit expressions of the terms on the RHS to control their limits. First:
\begin{align*}
    \psrreg{\partial^2 \liou (x_1) \mathbf{V}} &= \sum_{k=2}^N \frac{\im\alpha_k}{2(x_1-x_k) ^2} \psrreg{ \mathbf{V}}  -\frac{\im\mu\beta}{2} \int_{\C_\rho} \frac{1}{(x_1-z) ^2} \psrreg{V_\beta(z) \mathbf{V}} |\dd z|^2
\end{align*}
By Taylor-expanding $z\mapsto\psrreg{V_\beta(z) \mathbf{V}}$ in a neighborhood of $x_1$ (at order $\min(2,2+\beta(\beta+\alpha_1))$), the arguments used to define the $\L_{-2}V_{\alpha_1}$ descendant show that this expression admits a well-defined limit as $\rho\to0$. For the second term we have
\begin{align*}
    &\psrreg{(\partial \liou(x_1))^2 \mathbf{V}}= -\sum_{k,l=2}^N \frac{\alpha_k\alpha_l}{4(x_1-x_k)  (x_1-x_l) } \psrreg{\mathbf{V}} \\
    &\qquad +\mu \int_{\C_\rho} \left( \frac{ \beta^2 }{4(x_1-z) ^2} + \sum_{k=2}^N \frac{ \beta \alpha_k}{2(x_1-x_k) (x_1-z) } \right) \psrreg{V_\beta(z) \mathbf{V}} |\dd z|^2 \\
    &\qquad -\mu^2 \int_{\C_\rho} \int_{\C_\rho} \frac{\beta^2}{4(x_1-z) (x_1-z^\prime) } \psrreg{V_\beta(z)V_\beta(z^\prime)\mathbf{V}} |\dd z|^2 |\dd z^\prime|^2.
\end{align*}
(integrations in the last term may be exchanged by Fubini). Hence
\begin{align*}
    &\partial^2_{x_1} \psrreg{ \mathbf{V} } =  \left( \sum_{k=2}^N \frac{-\alpha_1 \alpha_k}{2(x_1-x_k)^2 } + \sum_{k,l=2}^N \frac{\alpha_1^2\alpha_k\alpha_l}{4(x_1-x_k)  (x_1-x_l) } \right) \psrreg{\mathbf{V}} \\
    &-\mu\int_{\C_\rho}\left( \sum_{k=2}^N\frac{ \alpha_1^2 \alpha_k\beta}{2(x_1-x_k)(x_1-z) } -\frac{\alpha_1\beta (1-\frac{\beta\alpha_1}2)}{2(x_1-z)^2 }\right)\psrreg{V_\beta(z)\mathbf{V}} \; |\dd z|^2\\
    &\quad+\mu^2 \int_{\C_\rho} \int_{\C_\rho} \frac{\alpha_1^2\beta^2}{4(x_1-z) (x_1-z^\prime) } \psrreg{V_\beta(z)V_\beta(z^\prime)\mathbf{V}} |\dd z|^2 |\dd z^\prime|^2.
\end{align*}
Like before, we remove the double integral by relating it to the integral of a total derivative:
\begin{align*}
    \partial_z \left(\frac{1}{(x_1-z) } \psrreg{V_\beta(z) \mathbf{V} } \right) = &\frac{1-\frac{\beta\alpha_1}{2}}{(x_1-z) ^2} \psrreg{V_\beta(z) \mathbf{V} }  - \sum_{k=2}^N \frac{\beta \alpha_k}{2(x_1-z) (x_k-z) } \psrreg{V_\beta(z)\mathbf{V}}\\
    &-\mu\int_{\C_\rho} \frac{\beta^2}{2(x_1-z) (z-z^\prime) } \psrreg{V_\beta(z)V_\beta(z^\prime)  \mathbf{V}} |\dd z^\prime|^2
\end{align*}
 Before plugging the resulting expression for the two-fold integral in Eq.~\eqref{eq:double_derivative}, symmetrization identities are used:
\begin{align*}
    &\frac{1}{(x_1-z)(z-z^\prime)} = \frac{1}{x_1-z^\prime} \left(\frac{1}{x_1-z}+\frac{1}{z-z^\prime}\right)\\
    &\frac{1}{(x_1-z)(x_k-z)} = \frac{1}{x_k-x_1} \left(\frac{1}{x_1-z}-\frac{1}{x_k-z}\right) \ .
\end{align*}

Wrapping up, we find 
\begin{eqs}
\label{eq:double_derivative}
    \partial^2_{x_1} \psrreg{\mathbf{V} } &=  \alpha_1^2\left( \sum_{k=2}^N \frac{\left(\frac{\alpha_k}2+\alpha_1-\alpha_1^{-1}\right)\frac{\alpha_k}2}{(x_1-x_k)^2} + \sum_{k=2}^N\frac{1}{(x_1-x_k)}\sum_{\substack{l=1\\l\neq k}}^N \frac{\alpha_k\alpha_l}{2(x_k-x_l) } \right) \psrreg{\mathbf{V}}  \\
    &\quad -\alpha_1^2 \mu\int_{\C_\rho} \sum_{k=2}^N\frac{ \alpha_k\beta}{2(x_1-x_k)(x_k-z) } \psrreg{V_\beta(z)\mathbf{V}} \; |\dd z|^2\\
    &\quad- \mu \int_{\C_\rho} \partial_z\left(\frac{ \alpha_1^2}{(x_1-z) } \psrreg{V_\beta(z)\mathbf{V}}\right) |\dd z|^2\\
    &\quad+\mu\frac{ \alpha_1\beta}{2} \left(\frac{2}{\beta}-\alpha_1\right)\left(\frac{\beta}{2}+\alpha_1\right) \int_{\C_\rho}\frac{1}{(x_1-z)^2} \psrreg{V_\beta(z)\mathbf{V}} |\dd z|^2.
\end{eqs}

We now check whether each of those terms stays finite as $\rho\to 0$. The terms on the first line of Eq.~\eqref{eq:double_derivative} clearly do. Regarding the total derivative, the argument is the same as in the proof of Thm.~\ref{thm:desc}. The matter is more delicate for the last integral. Similarly to the proof of Prop.~\ref{prop:derivability}, we focus on the contribution of the domain $A_{\eta,\rho}$ where divergences may occur. Following Lemma~\ref{lemma:fusion_general} and our hypothesis $\alpha_1 > -\frac{1}{\beta}-\frac{\beta}{2}$, we expand the integrand as a power series in $(x_1-z)$. Due to cancellation of angular integrals, we have
\begin{eqs}
    \int_{A_{\eta,\rho}} \frac{1}{(x_1-z)^2} \psrreg{V_\beta(z)\mathbf{V}} |\dd z|^2 = \int_{A_{\eta,\rho}} o(|x_1-z|^{\beta \alpha_1 + \min(\beta(\alpha_1+\beta),0)} )|\dd z|^2
\end{eqs}
This integral converges by the bound on $\alpha_1$, concluding the proof.
\end{proof}

\begin{remark}
    Our proof shows that for general $\alpha_1$, in the sense of weak derivatives,
    \begin{eqs}
        \partial^2_{x_1} \ps{\mathbf{V} } &=  \alpha_1^2\left( \sum_{k=2}^N \frac{\left(\frac{\alpha_k}2+\alpha_1-\alpha_1^{-1}\right)\frac{\alpha_k}2}{(x_1-x_k)^2} + \sum_{k=2}^N\frac{\partial_{x_k}}{(x_1-x_k)}\right) \ps{\mathbf{V}} \\
    &\quad+\mu\frac{ \alpha_1\beta}{2} \left(\frac{2}{\beta}-\alpha_1\right)\left(\frac{\beta}{2}+\alpha_1\right) \int_{\C}\frac{1}{(x_1-z)^2} \ps{V_\beta(z)\mathbf{V}} |\dd z|^2.
    \end{eqs}
    The special values $\alpha_1\in\{0,-\frac{\beta}{2},\frac{2}{\beta}\}$ that cancel the last integral in Equation~\eqref{eq:double_derivative} precisely correspond to the primary fields that are degenerate at level 1 and 2 ($\alpha_1=\alpha_{r,s}$ with $rs\leq 2$ where the $\alpha_{rs}$ are the weights from the Kac table).
\end{remark}

\begin{proof}[Proof of Thm.~\ref{thm:BPZ}]
This follows from Thm.~\ref{thm:desc} and Prop.~\ref{prop:second_derivative}: since 
\begin{align*}
\L_{(-1,-1)}^\alpha[f] &= \i\alpha \partial^2f - \alpha^2 (\partial f)^2,\quad \L_{-2}^\alpha[f] = \i(Q+\alpha) \partial^2 f - (\partial f)^2
\end{align*}
we see that the linear combination $-\frac{1}{\alpha^2}\L_{(-1,-1)}^\alpha[f] + \L_{-2}^\alpha[f]$ vanishes when $\alpha\in\{-\frac{\beta}{2},\frac{2}{\beta}\}$. Hence at the level of regularized correlation functions,
\[
    \psrreg{\left(-\frac{1}{\alpha^2}(\L_{-1})^2 + \L_{-2}\right) V_\alpha(x) \mathbf{V}} = 0 \qt{if} \alpha\in\left\{-\frac{\beta}{2},\frac{2}{\beta}\right\} \ .
\]
Now, Thm.~\ref{thm:desc} and Prop.~\ref{prop:second_derivative} ensure that these functions are pointwise converging as $\rho\to0$. They also give expressions for the limits, which give rise to the BPZ equation.
\end{proof}


\appendix

\section{On the Laplace transform}\label{appendix:laplace}

Let $g$ be a metric on $\C$, $U\subseteq \C$ be an open set. Let $h$ be a continuous function and let $f(\bm{x},y)=\prod_{j=1}^N |x_j-y|^{\beta \alpha_k} e^{2h(y)}e^{2h(x_k)}$ such that $f$ is bounded at infinity\footnote{$h$ corresponds to the Robin mass.}. Let $u$ be a bounded function on $U$. Put
\begin{eqs}
    \mathcal{G}(\mu,fe^u,\bm{x}) \coloneqq \ga{e^{-\mu M^g_\beta(U,fe^u)}}
\end{eqs}
The algebraic decay of $\mathcal{G}(\mu,f,\mathrm{x})$ may be understood by arguments about the pole structure of its Mellin transform that we now explain. The Laplace transform of the imaginary GMC is the generating function of the positive-order moments:
\begin{eqs}
\label{eq:ms_def}
    \mathcal{G}(\mu,fe^u,\bm{x}) &= \sum_{s=0}^{+\infty} \frac{(-\mu)^s}{s!} \lim\limits_{\eps\to0} \ga{\left(\int_{U} f(y,\bm{x})e^{u(y)} M^g_{\beta,\eps}(\dd^2 x)\right)^s} \\
    &=: \sum_{s=0}^{+\infty} \frac{(-\mu)^s}{s!} m_s(fe^u,\bm{x})
\end{eqs}
where the limit $\eps\to 0$ can be taken on each moment because of the existence of exponential moments and dominated convergence on integrals: for any $s\in\N$,
\begin{eqs}
\label{eq:moments}
    &m_s(fe^u,\mathrm{x}) = \int_{\C^s} \prod_{k=1}^s f(y_k) e^{u(y_k)} \prod_{k<l} \left(|y_k-y_l|e^{\frac{1}{2}(W_g(y_k)+W_g(y_l)} \right)^{\beta^2} \dd v_g(y_k) \ 
\end{eqs}
is a convergent integral. For instance, in the hyperbolic (round) metric for $f(\bm{x},y) = f_0 = \prod_{j=1}^N e^{-\beta \alpha_j G_{g_0}(x_k,y)}$ and $u=0$,
\begin{eqs}
\label{eq:moments_round_metric}
    &m_s(f_0,\mathrm{x}) = \int_{\C^s} \prod_{k=1}^s \prod_{j=1}^N \left(|y_k-x_j|g_0(y_k)^\frac{1}{4}g_0(y_j)^\frac{1}{4}\right)^{\beta\alpha_j} \\
    &\qquad \qquad \qquad \qquad \qquad \quad \times \prod_{k<l} \left(|y_k-y_l|g_0(y_k)^\frac{1}{4}g_0(z_l)^\frac{1}{4}\right)^{\beta^2} g_0(z_k)\dd^2 y_k \ .
\end{eqs}
We now turn to study the regularity of $m_s$ as a function of the charges $(\alpha_j)$, $\beta$ seen as complex numbers. It is easily seen from keeping track of collisions between $s$ integration variables that a necessary condition for convergence is 
\begin{eqs}
\label{eq:condition1}
\Re \beta^2 > -\frac{4}{s} \ .
\end{eqs}
Furthermore, collision around $x_k$ leads to
\begin{eqs}
\label{eq:condition2}
\Re \left[\frac{2(\alpha_j-Q)}{\beta}\right] > -s \ .
\end{eqs}
Convergence of the integrals at infinity is granted when insertions are in a compact subset of the complex plane because of the Robin masses and the hypothesis on $f$ and $u$. However, if we allow an insertion point $x_k$ to be taken to $\infty$ as in Conjecture~\ref{conj:Laplace2}, $|y-x_k| e^{2h(y)} e^{2h(x_k)} = \mathcal{O}((1+|y|)^{-\beta \alpha_k})$. A necessary condition for convergence then is
\begin{eqs}
\label{eq:condition3}
    \Re\left[\beta \alpha_j\right] > -2 \ .
\end{eqs}

\textit{Suppose} that there exists an analytic continuation of $m_s(fe^u,\bm{x})$ to $s\in\C$ that is analytic in the half-plane $\Re s > \omega$ for some $\omega \in \R_+$, and subexponential in every direction away from possible accumulation line of poles ($|m_s|\leq C e^{\tau |s|}$ for some $C\geq 0$, $\tau \in \R$ and $\arg s \neq \theta_0$ for some $\theta_0$). Then the Laplace transform is equal to the inverse Mellin transform
\begin{eqs}
\label{eq:mellin}
    \mathcal{G}(\mu,f,\bm{x}) = \int_{\omega+i\R} \Gamma(-s) m_s(f,\mathrm{x}) \dd s
\end{eqs}
where $\omega <0$ belongs to a neighborhood of the origin where the integrand is holomorphic. The equality between Eq.~\eqref{eq:ms_def} and~\eqref{eq:mellin} is apparent from closing the integration contour to surround poles in the \emph{rightmost} half plane. However, due to the subexponential decay, one could also close the integration contour to the \emph{left}. From conditions~\eqref{eq:condition1} and~\eqref{eq:condition2}, we conjecture that the analytic continuation of $m_s$ has poles at $s=-2\Re[(\alpha_j-Q)/\beta]$, $s=-\Re[4/\beta^2]$. These poles are all located on the negative real line if the probabilistic constraints $\alpha_j > Q$, $\beta\in(0,\sqrt{2})$ hold.

Following this line of thought, a simple residue calculation shows that the Laplace transform decays algebraically as soon as the first Seiberg bound $\alpha_j > Q$ holds, with exponent given by the first pole on the left of $\omega$\footnote{Our prediction differs slightly from the exact results on the Laplace transform of the Imaginary GMC on the circle~\cite{usciatiProbabilisticConstructionNoncompactified2026}. Indeed, in the latter setting, there cannot be any collision of integration variables around the insertion point $z=0$, so the part of the bound depending on the charges $\alpha_j$ does not appear. Nevertheless, an exact calculation presented in the aforementioned article shows that the exponent is indeed given by the \textit{second} part of our bound.}
\begin{align*}
\lambda = \min\limits_{j=1,\dots,N} \frac{2(\alpha_j-Q)}{\beta} \wedge \frac{4}{\beta^2} \ .
\end{align*}
Note that when $\alpha_j=\frac{\beta}{2}$ for some $j$, poles collide and become of order two. This is coherent with the divergence of the conjectured three-point function of Imaginary Liouville (the Imaginary DOZZ formula) at this particular value.

All in all, our guesses lead us to
\begin{eqs}
    \sum_{\mu \geq 1} \sum_{\|u\|_\infty \leq 1} \mu^\lambda \norm{\mathcal{G}(\mu,fe^u,\bm{x})} < \infty \ .
\end{eqs}
The decay exponent is not modified as long as insertion points stay separated. We believe that the bound is uniform, because the dependency of moments on $\bm{x}$ is gentle enough.

\bibliography{biblio_}
\bibliographystyle{alpha}
    
\end{document}